\documentclass[11pt]{article}
\usepackage[utf8]{inputenc}
\usepackage[USenglish]{babel}
\usepackage{hyperref}
\usepackage{amsmath,amssymb,amsfonts,amsthm}
\usepackage{mathabx}
\usepackage{url}
\usepackage{fullpage}
\usepackage{cleveref}
\usepackage[usenames,dvipsnames,table]{xcolor}
\usepackage{graphicx}
\usepackage{booktabs}
\usepackage{authblk}
\usepackage{todonotes}
\usepackage{comment}
\usepackage{enumerate}
\usepackage{enumitem}
\usepackage{authblk}
\usepackage{scalerel}
\usepackage{thm-restate}
\usepackage[title]{appendix}
\usepackage{soul}
\usepackage{nicefrac}
\usepackage{thm-restate}
\usepackage{listings}
\usepackage[multiple]{footmisc}

\newtheorem{theorem}{Theorem}
\newtheorem{lemma}[theorem]{Lemma}
\newtheorem{corollary}[theorem]{Corollary}
\theoremstyle{definition}
\newtheorem{definition}{Definition}
\newtheorem{observation}[theorem]{Observation}

\newcommand{\conditioned}{\text{ }\Big\vert\text{ }}

\newcommand{\calR}{\mathcal R}

\newcommand{\calE}{\mathcal E}

\newcommand{\calA}{\mathcal A}
\newcommand{\calB}{\mathcal B}
\newcommand{\calC}{\mathcal C}
\newcommand{\calD}{\mathcal D}
\newcommand{\calY}{\mathcal Y}
\newcommand{\derive}[1]{\widetilde {#1}}
\newcommand{\toptab}[1]{\widehat {#1}}

\newcommand{\hstar}{h^{*}}

\newcommand{\cond}{\;|\;}

\newcommand{\eps}{\varepsilon}

\newcommand{\E}{\textnormal{E}}

\newcommand{\sd}{\triangle}
\newcommand{\xor}{\oplus}

\newcommand{\Prp}[1]{\Pr\!\left[{#1} \right]}

\newcommand{\Prpcond}[2]{\Pr\!\left[{#1} \mid {#2} \right]}

\newcommand{\Ep}[1]{\E\!\left[{#1} \right]}
\newcommand{\Epcond}[2]{\E\!\left[{#1} \mid {#2} \right]}

\newcommand{\indicator}[1]{\left[{#1}\right]}

\newcommand{\abs}[1]{\left | #1 \right |}
\newcommand{\set}[1]{\left \{ #1 \right \}}

\crefname{lemma}{Lemma}{Lemmas}
\Crefname{lemma}{Lemma}{Lemmas}
\crefname{theorem}{Theorem}{Theorems}
\Crefname{theorem}{Theorem}{Theorems}
\Crefname{corollary}{Corollary}{Corollaries}
\crefname{corollary}{Corollary}{Corollaries}
\crefname{observation}{Observation}{Observations}
\Crefname{observation}{Observation}{Observations}
\crefname{definition}{Definition}{Definitions}
\Crefname{definition}{Definition}{Definitions}
\crefname{section}{Section}{Sections}
\Crefname{section}{Section}{Sections}
\crefname{figure}{Figure}{Figures}
\Crefname{figure}{Figure}{Figures}
\crefname{appendix}{Appendix}{Appendices}
\Crefname{appendix}{Appendix}{Appendices}

\def\ul{\underline}
\def\cR{\mathcal R}

\def\cU{\mathcal U}

\def\bF{\mathbb F}

\DeclareMathOperator*{\bigsd}{\scalerel*{\triangle}{\sum}}
\newcommand{\size}[1]{\ensuremath{\left|#1\right|}}
\newcommand{\ld}{\left}
\newcommand{\rd}{\right}

\newcommand{\parentheses}[1]{\left(#1\right)}
\newcommand{\ssigma}{|\Sigma|}
\newcommand\req[1]{(\ref{#1})}
\newcommand\ull[1]{\underline{\underline{#1}}}

\def\DiffKeys{\textnormal{DiffKeys}}
\newcommand{\DepProb}{\textsf{DependenceProb}}
\newcommand{\DP}{7\mu^3(3/|\Sigma|)^{d+1}+1/2^{|\Sigma|/2}}
\newcommand\drop[1]{}

\title{Locally Uniform Hashing}
\author[1]{Ioana O. Bercea}
\author[2]{Lorenzo Beretta}
\author[2]{Jonas Klausen}
\author[2]{Jakob Bæk Tejs Houen}
\author[2]{Mikkel Thorup}
\affil[1]{IT University of Copenhagen \\ \tt{\{iobe\}@itu.dk}}
\affil[2]{University of Copenhagen\\
\tt{\{beretta, jokl, jakn, mthorup\}@di.ku.dk}}
\date{}

\begin{document}
\setcounter{page}{0}
\maketitle
\begin{abstract}
	
Hashing is a common technique used in data processing, with a  strong
impact on the time and resources spent on computation. 
Hashing also affects the applicability of theoretical results that often assume access to (unrealistic) uniform/fully-random hash functions. In this paper, we are concerned with designing hash functions that are practical and come with strong theoretical guarantees on their performance.

To this end, we present
tornado tabulation hashing, which is simple, fast, and exhibits a certain full, local randomness property that provably makes diverse
algorithms perform almost as if (abstract) fully-random hashing was
used. For example, this includes classic linear probing, the
widely used HyperLogLog algorithm of  Flajolet, Fusy, Gandouet,
Meunier [AOFA’97] for counting
distinct elements, and the one-permutation hashing of Li, Owen, and Zhang
[NIPS’12] for large-scale machine learning. We also provide a
very efficient solution for the classical problem of obtaining fully-random hashing on a fixed (but unknown to the hash function) set of $n$ keys using $O(n)$ space. As a consequence, we get more efficient implementations
of the splitting trick of Dietzfelbinger and Rink [ICALP'09]
and the succinct space uniform hashing of Pagh and Pagh [SICOMP'08].

Tornado tabulation hashing is based on a simple method to systematically break dependencies in tabulation-based hashing techniques.
\end{abstract}

\thispagestyle{empty}
\newpage
\setcounter{page}{1}
\section{Introduction}
The generic goal of this paper is to create a practical hash function
that provably makes important algorithms behave almost as if
(unrealistic) fully random hashing was used. By practical,
we mean both simple to implement and fast. Speed is
important because hashing is a common inner loop for data processing.
Suppose for example that we want to sketch a high-volume data stream such
as the packets passing a high-end Internet router. If we are too slow,
then we cannot handle the stream at all. Speed matters, also within constant
factors.

The use of weak hash functions is dangerous, not only in theory but also
in practice. A good example is the use of classic linear hashing. Linear
hashing is 2-independent and
Mitzenmacher and Vadhan \cite{mitzenmacher08hash} have proved that, for some applications, 2-independent hashing
performs as well as fully random hashing if the input has enough entropy, and indeed
this often works in practice. However, a dense set has only 1 bit of entropy per element,
and \cite{thorup12kwise,patrascu10kwise-lb} have shown that with a linear hashing scheme, if the input is a dense set (or more generally, a dense subset of an arithmetic sequence), then linear probing
becomes extremely unreliable and the expected probe length\footnote{The probe length is defined as the number of contiguous cells probed to answer a query of a linear probing hash table.} increases from constant
to $\Omega(\log n)$. This is also a practical problem because dense subsets may occur for
many reasons. However, if the system is only tested on random inputs, then we may not discover
the problem before deployment.

The issue becomes even bigger with, say, sampling and estimation where we typically just trust our estimates with no knowledge of the true value. We may never
find out that our estimates are bad. With 2-independent hashing, we get the right variance, but not exponential concentration. Large errors can happen way too
often, yet not often enough to show up in a few tests. This phenomenon is demonstrated on synthetic data in \cite{aamand2020fast}
and on real-world data in \cite{Aamand0KKRT22}. All this shows the
danger of relying on weak hash functions without theoretical guarantees
for all possible inputs, particularly for online systems where we
cannot just change the hash function if the input is bad, or in situations
with estimates whose quality cannot be verified. One could instead, perhaps, use hash functions based on cryptographic assumptions, but the hash function that we propose
here is simple to implement, fast, and comes with strong unconditional guarantees.

In this paper, we introduce \emph{tornado tabulation hashing}. A
tornado tabulation hash function $h:\Sigma^c\to \calR$ requires
$O(c\size{\Sigma})$ space and can be evaluated in $O(c)$ time using, say,
$2c$ lookups in tables with $\abs{\Sigma}$ entries plus some simple
AC$^0$ operations (shifts, bit-wise xor, and assignments). As with
other tabulation schemes, this is very fast when $\Sigma$ is small
enough to fit in fast cache, e.g., for 32-bit keys divided into $c=4$ 
characters of 8 bits (namely, $|\Sigma| = 2^8$), the speed is similar to that of evaluating a degree-2 polynomial over a Mersenne prime field.

Tornado hashing has the strong property that if we hash a set of
$\abs{\Sigma}/2$ keys, then with high probability, the hash values are
completely independent. For when we want to handle many more keys, e.g., say $\abs{\Sigma}^3$ (as is often the case when $\Sigma$ is small),  tornado tabulation
hashing offers a certain \emph{local uniformity} that provably makes a
diverse set of algorithms behave almost as if the hashing was fully random on all the keys.  The definition of local uniformity is due to Dahlgaard, Knudsen,
Rotenberg, and Thorup~\cite{dahlgaard15k-partitions}. The definition is a bit complicated,
but they demonstrate how it
applies to the widely used HyperLogLog algorithm of the
Flajolet, Fusy, Gandouet,
Meunier~\cite{Flajolet07hyperloglog}
for counting distinct elements in data streams, and the One-Permutation Hashing of Li, Owen, and Zhang ~\cite{li12oneperm} used
for fast set similarity. They 
conclude that the estimates
obtained are only a factor $1+o(1)$ worse than if fully-random hashing was used. Interestingly, \cite{dahlgaard15k-partitions} proves this 
in a high-level black-box manner. Loosely speaking, the point is that the algorithm using
locally uniform hashing behaves
as well as the same algorithm using fully-random hashing on a slightly  worse input. 

As a new example, we will demonstrate
this on linear probing.
Knuth's original 1963 analysis of linear
probing~\cite{knuth63linprobe}, which started the field of algorithms
analysis, showed that with fully-random hashing and load $(1-\eps)$, the expected probe length
is $(1+1/\eps^2)/2$. From this, we conclude that tornado tabulation hashing yields
expected probe length $(1+o(1))(1+1/\eps^2)/2$, and we get that without having to reconsider Knuth's analysis.

For contrast, consider the work on linear probing with $k$-independent hashing. Pagh, Pagh, Ru{\v z}i{\'c}~\cite{pagh07linprobe} showed that $5$-independence is enough to obtain a bound of $O(1/\eps^{13/6})$ on the expected  probe length. This
was further improved to $O(1/\eps^2)$ by P{\v a}tra{\c s}cu  and Thorup~\cite{patrascu12charhash}, who achieved the optimal $O(1/\eps^2)$. They matched Knuth's bound modulo some
large constants hidden in the $O$-notation and needed a very different
analysis.

In practice, the guarantee that we perform almost like fully-random hashing means
that no set of input keys will lead to substantially different performance statistics.
Thus, if we tested an online system on an arbitrary set of input keys, then we do not
have to worry that future input keys will degrade the performance statistics, not even
by a constant factor.

The definition of local uniformity is due to Dahlgaard, Knudsen,
Rotenberg, and Thorup~\cite{dahlgaard15k-partitions}. They did not name it as
an independent property, but they described it
as a property of their new hashing
scheme: mixed tabulation hashing.
However, the local uniformity
of mixed tabulation 
assumes table size $\abs{\Sigma}\to\infty$, but the speed of
tabulation hashing relies on $\abs{\Sigma}$ being small enough to fit
in fast cache and all reported experiments use 8-bit characters (see
\cite{patrascu12charhash,aamand2020fast,Aamand0KKRT22,DahlgaardKT17:nips}). However, none of the bounds
from \cite{dahlgaard15k-partitions} 
apply to 8 or even 16-bit characters,
e.g., they assume $O(\log\abs{\Sigma})^c<\abs{\Sigma}$. Our new scheme avoids the
exponential dependence on $c$, and we get explicit error probability
bounds that are meaningful, not 
just in theory, but also for 
practice with tables in fast cache. 

For when we want full randomness on more keys than fit in fast cache, we could, as above, increase $\abs{\Sigma}$ in all
$O(c)$ lookup tables. In this paper, we show that it suffices to
augment the in-cache tornado hashing with just $2$ lookups in tables of size $2n$
to get full randomness on $n$ keys with high probability. This would
work perfectly inside a linear space algorithm assuming fully-random hashing,
but it also leads to more efficient implementations
of the spitting trick of Dietzfelbinger and Rink \cite{dietzfel09splitting}
and the succinct space uniform hashing of Pagh and Pagh \cite{PP08}.

In \Cref{sec:tornado-tabulation-definition} we define our hash function, and in \Cref{sec:technical-results} we present our technical results, including the definition of local uniformity. In~\Cref{sec:intro-related-work}, we discuss more explicitly how our work compares to mixed tabulation and explain some of our techniques in comparison. In \Cref{sec:intro-local-power} we discuss several applications. In~\Cref{sec:intro-high}, we relate our work to previous work on achieving highly independent hash functions. Finally, in \Cref{sec:splitnshare} we discuss how tornado tabulation can be employed to improve the so-called ``splitting trick'' and succinct uniform hashing.

\subsection{Tornado tabulation hashing} \label{sec:tornado-tabulation-definition}

\paragraph{Simple tabulation hashing.} We first introduce our main building block, which is the simple tabulation hash function dating back to at least Zobrist~\cite{zobrist70hashing} and Wegman and Carter~\cite{wegman81kwise}.
Throughout the paper, we will consider keys to come from the universe $\Sigma^c$ and hash values to be in $\cR = [2^r]$. More concretely, we interpret a key $x$ as the concatenation of $c$ \emph{characters} $x_1 \dots x_c$ from $\Sigma$. We then say that a function $h: \Sigma^c \longrightarrow \cR$ is a \emph{simple tabulation} hash function if 
\[
h(x) = T_1[x_1] \xor \cdots \xor T_c[x_c]
\]
where, for each $i = 1 \dots c$, $T_i: \Sigma \longrightarrow \cR$ is a fully-random function stored as a table. 

We think of $c$ as a small constant, e.g., $c=4$, for 32-bit keys divided into 8-bit characters,
yet we will make the dependence on $c$ explicit.
We assume that both keys and hash values fit in a single word and that $\abs{\Sigma}\geq 2^8$.

\paragraph{Tornado tabulation hashing.} 

To define a tornado tabulation hash function
$h$, we use several simple tabulation hash functions.  A
tornado tabulation function has
a number $d$ of \emph{derived} characters. Think of
$d$ as, say, $c$ or $2c$. It will later determine our 
error probability bounds. We will always assume $d=O(c)$ so that $d$ characters from $\Sigma$ can be represented in $O(1)$ words of memory (since a key from $\Sigma^c$ fits in a single word).

For each $i = 0,\ldots, d$, we let $\derive h_i : \Sigma^{ c+i - 1} \longrightarrow \Sigma$ be a simple tabulation hash function. Given a key $x \in \Sigma^c$, we define its \emph{derived key} $\derive h(x) \in \Sigma^{c+d}$ as $\derive x=\derive x_1\cdots \derive x_{c+d} $, where
\begin{equation}\label{eq:tornado}
	\derive x_i = \begin{cases}
		x_i &\text{if $i<c$} \\
		x_{c} \xor \derive h_{0}(\derive x_1 \cdots \derive x_{c-1}) &\text{if $i = c$} \\
		\derive h_{i-c}\ld(\derive x_1 \cdots \derive x_{i-1}\rd) &\text{if $i > c$}.
	\end{cases}
\end{equation}%
We note that each of the $d$ derived characters $\derive x_{c+1}, \ldots, \derive x_{c+d} $ is progressively defined by applying a simple tabulation hash function to all its preceding characters in the derived key $\derive x$. Hence, the name tornado tabulation. The step by which we obtain $\derive x_c$ corresponds to the twist from~\cite{PT13:twist}. By Observation $1.1.$ in~\cite{PT13:twist}, $x_1 \dots x_c \mapsto \derive x_1 \dots \derive x_c$ is a permutation, so distinct keys
have distinct derived keys.
Finally, we have a simple tabulation hash function $\widehat h: \Sigma^{c+d} \longrightarrow \cR$, that we apply to the derived key.
The \emph{tornado tabulation} hash
function $h: \Sigma^c \longrightarrow \cR$ is then
defined as $h(x) = \widehat h(\derive x)$.

\paragraph{Implementation.} The simplicity of tornado tabulation is apparent from its C implementation below. In the code, we fold tornado's lookup tables together so we
can implement them using
$c+d$ character tables $\Sigma\to \Sigma^{d+1}\times\cR$.  Elements of $\Sigma^{d+1}\times\cR$ are just
represented as $w$-bits
numbers. 
For memory alignment and
ease of implementation,
we want $w$ to be a power of
two such as $64$ or $128$.

We now present a C-code implementation of tornado tabulation for 32-bit keys, with \(\Sigma = [2^{8}]\), \(c=4\), \(d=4\), and \(\cR=[2^{24}]\).
Besides the key \texttt{x}, the function takes as input an array of \(c+d\) tables of size \(\ssigma\), all filled with independently drawn 64-bit values. 
\drop{
	The 64-bit entries allow us to perform all lookups for each character \(x_i\) or \(\derive x_i\) just once, and xor it onto the preliminary result \texttt{h} which also contains the derived characters and the twisted \(\derive x_c\).
	Note that the key is read in reverse, such that \(x_1\) corresponds to the least significant bits and \(x_c\) is \texttt{x >> 24}.
	
	The first loop takes care of lookups related to the (unaltered) characters \(x_1, \dots, x_{c-1}\) while the lookups based on \(\derive x_c, \dots, \derive x_{c+d}\) are handled in the second loop.
	For each iteration, the 8 remaining least significant bits are cast to a byte which is used for indexing into the appropriate table \texttt{H[i]}.}

\begin{lstlisting}[language=C,basicstyle=\footnotesize]
	INT32 Tornado(INT32 x, INT64[8][256] H) {
		INT32 i; INT64 h=0; INT8 c;
		for (i=0;i<3;i++) {
			c=x;
			x>>=8;
			h^=H[i][c];}
		h^=x;
		for (i=3;i<8;i++) {
			c=h;
			h>>=8;
			h^=H[i][c];} 
		return ((INT32) h);}
\end{lstlisting}

\paragraph{Speed.}
As we can see in the above
implementation, tornado
hashing uses $c+d$
lookups and $O(c+d)$ simple AC$^0$
operations. 
The speed of tabulation hashing
depends on the tables
fitting in fast cache which
means that $\Sigma$ should not
be too big. In the above
code, we used $\Sigma=[2^8]$, as in all
previously reported experiments
with tabulation hashing. (see
\cite{patrascu12charhash,aamand2020fast,Aamand0KKRT22,DahlgaardKT17:nips}).

The speed of tabulation
schemes is incomparable to
that of polynomial methods
using small space but multiplication. Indeed, the ratio between the cost of cache lookups and multiplication depends on the architecture. In line
with previous experiments,
we found our tornado implementation
for 32-bit keys to be as fast
as a degree-2 polynomial over a Mersenne prime ($2^{89}-1$) field.

We note that our implementation
for the random table \texttt{H} only needs a pointer to an area
filled with ``fixed" random bits, and it could conceivably be made much faster if we instead of cache had access to random bits stored in simple read-only memory (ROM or EPROM).

\subsection{Theoretical Results} \label{sec:technical-results}
The main aim of our paper is to prove that, with high probability (whp), a tornado tabulation hash function is fully random on some set $X$ of keys. The challenge is to characterize for which kinds of sets we can show such bounds.

\medskip\noindent
\textbf{Full randomness for fixed keys.} We begin with a simpler result that follows directly from our main technical theorem. In this case, the set $X$ of keys is fixed. 

\begin{theorem}\label{thm:intro-fixed-set}
	Let $h:\Sigma^c\to \cR$ be
	a random tornado tabulation hash 
	function with $d$ derived
	characters. For any
        fixed $X\subseteq \Sigma^c$, if  $\size{X}\leq \ssigma/2$,
	then $h$ is fully random on $X$ with probability at least
	$$1-7|X|^3(3/|\Sigma|)^{d+1}- 1/2^{\size{\Sigma}/2}\;.$$
\end{theorem}

With $c,d=O(1)$,~\Cref{thm:intro-fixed-set} gives an $O(\size{\Sigma})$ space hash function that can be evaluated in constant time and, with high probability, will be fully random for any fixed set $X$ of
at most $\abs{\Sigma}/2$ keys. This is asymptotically tight as
we need $|X|$ random hash values to get this randomness.

The random process behind the error probability that we get will be made clear in the next paragraph. We note here that,  since $|X|\leq|\Sigma|/2$, we have that $7|X|^3(3/|\Sigma|)^{d+1}\leq
24(3/|\Sigma|)^{d-2}$.
With $\abs{\Sigma}\geq 2^8$, the
bound is below $1/300$ for
$d=4$, and decreases rapidly
for larger $d$. For $c\geq 4$,
if we set $d=2c$, we get
an error probability below
$1/u$ where $u=\abs{\Sigma}^c$
is the size of the universe.
We can get error probability
$1/u^{\gamma}$ for any constant
$\gamma$ with $d=O(c)$, justifying
this assumption on $d$.

\medskip\noindent
\textbf{Linear independence.} The general structure of our results is to identify some error event such that~(1) if the event does not occur, then the hash function will be fully random on $X$, and~(2) the error event itself happens with low probability. The error event that we consider in~\Cref{thm:intro-fixed-set} is inherent to all tabulation-based hashing schemes.  
Namely, consider some set $Y$ of keys from some universe $\Sigma^b$.
We say that $Y$ is \emph{linearly independent} if and only if, for every subset $Y'\subseteq Y$, there exists a character position $i\in\set{1,\ldots, b}$ such that some character appears an odd number of times in position $i$ among the keys in $Y'$.
A useful connection between this notion and tabulation-based hashing was shown by Thorup and Zhang  \cite{thorup12kwise}, who proved that a set of keys is linearly independent if
and only if simple tabulation hashing is fully random on these keys:

\begin{lemma}[Simple tabulation on linearly independent keys] \label{lem:simple-tab-on-lin-indep-sets-regular}
	Given a set of keys $Y\subseteq \Sigma^b$ and a simple tabulation hash function $h: \Sigma^b \rightarrow \cR$, the following are equivalent:
	\begin{enumerate}[label = (\roman*)]
		\item $Y$ is linearly independent
		\item $h$ is fully random on $Y$ (i.e., $h|_Y$ is distributed uniformly over $\cR^{Y}$).\footnote{In general, we employ the notation $h|_S$ to denote the function $h$ restricted to the keys in some set $S$.}
	\end{enumerate}
\end{lemma}

To prove~\Cref{thm:intro-fixed-set}, we employ~\Cref{lem:simple-tab-on-lin-indep-sets-regular} with sets of derived keys.  Namely, given a set of keys $X \subseteq \Sigma^c$, we consider the error event that the set $\derive X = \set{\derive h(x) \mid x\in X}$ of its derived keys is linearly dependent. We then show that this happens with probability at most $7|X|^3(3/|\Sigma|)^{d+1}+1/2^{\size{\Sigma}/2}$. If this doesn't happen, then the derived keys are linearly independent, and, by~\Cref{lem:simple-tab-on-lin-indep-sets-regular}, we  get that the tornado tabulation hash function   $h = \toptab h \circ \derive h$ is fully random on $X$ since it applies the simple
tabulation hash function $\toptab h$ to the derived keys $\derive X$.
We note that the general idea of creating linearly independent lookups to create fully-random hashing goes back at least to 
\cite{siegel04hash}. The point of this paper is to do it in a really efficient way.

\paragraph{Query and selected keys.} Our main result,~\Cref{thm:intro-random-set}, is a more general version of~\Cref{thm:intro-fixed-set}.
Specifically, while~\Cref{thm:intro-fixed-set} holds for any fixed set of keys, it requires that $\size{X} \leq \size{\Sigma}/2$. For
a fast implementation, we want $\size\Sigma$ to be small
enough to fit in fast cache, e.g., $\size\Sigma=2^8$, but in most
applications, we want to hash
a set $S$ of keys that is much larger, e.g., $\size S\sim \size\Sigma^3$. Moreover, we might be interested in showing full randomness for subsets $X$ of $S$ that are not known in advance: consider, for instance, the set $X$ of all the keys in $S$ that hash near to $h(q)$ for some fixed key $q\in S$. In this case,~\Cref{thm:intro-fixed-set} would not help us, since the set $X$ depends on $h(q)$ hence on
$h$.

To model this kind of scenario, we consider a set of 
\emph{query keys} $Q \subseteq \Sigma^c$ and define a set of
\emph{selected keys} $X \subseteq \Sigma^c$. 
Whether a key  $x$ is selected or not depends only on $x$, its own hash value $h(x)$, and the hash values of the query keys $h|_Q$ (namely, conditioning on $h(x)$ and $h|_Q$ makes $x \in X$ and $h$ independent). In~\Cref{thm:intro-random-set}, we will show that, if the selected keys are few enough, then $h|_X$ is fully random with high probability. 

Formally, we have a selector function $f: \Sigma^c \times \cR \times \cR^Q \longrightarrow \{0, 1\}$ and we define the set of selected keys as
\[
X^{f,h} = \set{x \in \Sigma^c \cond f(x, h(x), h|_Q) = 1}.
\]
We make the special requirement that $f$  should always
select all the query keys $q\in Q$, that is, $f(q, \cdot, \cdot)=1$ regardless
of the two last arguments. We then define
\begin{equation}
	\label{eq:px-definition}
	\mu^f := \sum_{x\in \Sigma^c} p^f_x  \quad \text{ with } \quad
	p^f_x := \max_{\varphi \in \cR^Q} \Pr_{r \sim \cU(\cR)}\ld[f(x, r, \varphi)=1\rd]\;.
\end{equation}
Here the maximum is taken among all possible assignments of hash values to query keys $\varphi: Q \to \cR$ and $r$ is distributed uniformly over $\cR$.
Trivially we have that
\begin{observation}
	If $h:\Sigma^c\to\cR$ is fully random then $\E[|X^{f,h}|]\leq\mu^f$.
\end{observation}
When $f$ and $h$ are understood, we may omit these superscripts. It is
important that $X$ depends on both $f$ and $h$ while $\mu$ only
depends on the selector $f$. We now have the following main technical theorem:

\begin{restatable}{theorem}{maintechtheorem}\label{thm:tech-random-set}
	Let $h=\widehat h\circ \derive h:\Sigma^c\to\cR$
	be a random tornado tabulation hash function with $d$ derived characters and $f$ as described above. If $\mu^f \leq \Sigma / 2$
	then the derived selected keys $\derive h(X^{f,h})$ are linearly
	dependent with probability at most
	$\DepProb(\mu^f, d, \Sigma$), where
	$$\DepProb(\mu, d, \Sigma):=\DP\;.$$
\end{restatable}

Note that~\Cref{thm:tech-random-set} only bounds the probability of the error event.
Similarly as in~\Cref{thm:intro-fixed-set}, we would like to then argue that, if the error event does not happen,
we could apply~\Cref{lem:simple-tab-on-lin-indep-sets-regular} to claim that the final
hash values via the simple tabulation function $\widehat h$ are fully random. 
The challenge, however, is the presence of an inherent dependency in how the keys are selected to begin with, namely that $h=\widehat h\circ \derive h$ is already used to select the keys in $X^{f,h}$. In other words, by the time we want to apply $\widehat h$ to the derived selected keys $\derive h(X^{f,h})$, we have already used some information about $\widehat h$ in selecting them in $ X^{f,h}$.

\medskip\noindent
\textbf{Local uniformity.}
Nevertheless, there is a general type of selector functions for which we can employ~\Cref{thm:tech-random-set}  in conjunction with~\Cref{lem:simple-tab-on-lin-indep-sets-regular}  to claim full randomness. Namely, we consider selector functions that partition the bit representation of the final hash values into $s$ \emph{selection bits} and $t$ \emph{free bits} so that $\cR = \cR_s \times \cR_t = [2^s] \times [2^t]$. Given a key $x \in \Sigma^c$, we then denote by $h^{(s)}(x) \in \cR_s$ and $h^{(t)}(x) \in \cR_t$ the selection and free bits of $h(x)$ respectively. We then say that a selector function $f$ is an \emph{$s$-selector} if, for all $x\in \Sigma^c$, the output of $f(x, h(x), h|_Q)$ only depends on the selection bits of the hash function, i.e., $f(x, h(x), h|_Q) = f(x, h^{(s)}(x), h^{(s)}|_Q)$.

We now crucially exploit the fact that the output bits of a simple tabulation hash function are completely independent. 
Formally, we split the simple tabulation $\toptab h$ into two independent simple tabulation functions: $\toptab h^{(s)}$ producing the selection bits and $\toptab h^{(t)}$ producing the free bits. We then apply~\Cref{thm:tech-random-set} to $h^{(s)}=\toptab h^{(s)}\circ \derive h$
to conclude that the set of selected derived keys  $\derive h(X^{f,h^{(s)}})$
is linearly independent with high probability.  Assuming this, 
we then apply~\Cref{lem:simple-tab-on-lin-indep-sets-regular} to conclude that
$\toptab h^{(t)}$ is fully random on $\derive h(X^{f,\toptab h^{(s)}})$,
hence that  $h^{(t)}=\toptab h^{(t)}\circ \derive h$ is fully random on  $X^{f,h^{(s)}}$.

\begin{restatable}{theorem}{maintheorem}\label{thm:intro-random-set}
	Let $h:\Sigma^c\to \cR$ be a tornado tabulation hash  function with $d$ derived
	characters and $f$ be an $s$-selector as described above. If  $\mu^f \leq \Sigma / 2$,
	then $h^{(t)}$ is  fully random on $X^{f,h^{(s)}}$ with probability at least
	$$1-\DepProb(\mu^f, d, \Sigma)\;.$$
\end{restatable}

While the concept of an $s$-selector function might seem a bit cryptic, we note that it intuitively captures the notion of locality that linear probing and other applications depend on. Namely, the effect of the (high order\footnote{Thinking about selector bits as higher order bits helps our intuition. However, they do not have to be higher-order bits necessarily. More generally, any representation of $\cR$ as a product $\cR_s \times \cR_t$ would do the job.}) selector bits is to specify a dyadic interval\footnote{Recall that a dyadic interval is an interval of the form $[j2^i, (j+1)2^i)$, where $i,j$ are integers.} of a given length such that all the keys with hash values falling in that interval are possibly selected (with this selection further depending, perhaps, on the query keys $Q$, or on other specific selector bits, leading to more refined dyadic intervals).~\Cref{thm:intro-random-set} then says that the (low order) free bits of these selected keys will be fully-random with high probability. In other words, the distribution inside such a neighborhood is indistinguishable from what we would witness if we had used a fully-random hash function. 

As mentioned earlier, the concept of local uniformity stems from \cite{dahlgaard15k-partitions}, except that they did not consider
query keys. Also, they didn't name
the concept. They demonstrated
its power in different streaming
algorithms. For those applications,
it is important that \emph{the selection bits are not known to the algorithm}. They are only set in
the analysis based on the concrete
input to demonstrate good performance on this input. The problem in \cite{dahlgaard15k-partitions} is that their error probability bounds only apply when
the alphabet is so large that the tables do not fit in fast cache.
We will describe this issue closer in \Cref{sec:intro-related-work}.
In \Cref{sec:intro-local-power}
we will sketch the use
of local uniformity on linear probing where the locality
is relative to a query key.

\medskip\noindent
\textbf{Upper tail Chernoff bounds.}
\Cref{thm:intro-random-set} is only concerned with the distribution of the free bits of the selected keys, but to employ it in our applications, we often require that the number of selected keys is not much larger than $\mu^f$ with high probability (see~\Cref{sec:intro-local-power,sec:splitnshare}). We show that this size can be bounded from above with the usual Chernoff bound when the set of derived selected keys is linearly independent.

\begin{restatable}{lemma}{upperchernoff}\label{lemma:upper-chernoff}
	Let $h=\widehat h\circ \derive h:\Sigma^c\to\cR$
	be a random tornado tabulation hash function with $d$ derived characters, query keys $Q$ and selector
	function $f$.  Let $\mathcal{I}_{X^{f,h}}$ denote the event that the set of derived selected key $\derive h(X^{f,h})$ is linearly independent.
	Then, for any $\delta>0$,  the set $X^{f,h}$ of selected keys  satisfies the following:
	$$ \Pr\ld[{\size{X^{f,h}} \geq (1+\delta)\cdot  \mu^f \wedge \mathcal{I}_{X^{f,h}}  }\rd] \leq \parentheses{\frac{e^\delta}{(1+\delta)^{1+\delta}}}^{\mu^f}\;.$$
\end{restatable}
For a nice direct application,
consider hash tables with chaining, or throwing $n$ keys into $n$ bins
using tornado hashing. We can select a given bin, or the bin of a given query key. In either case
we have $\mu^f=1$ and then 
\Cref{lemma:upper-chernoff} together with \Cref{thm:tech-random-set} says that the
probability of getting $k$
keys is bounded by
\begin{equation}\label{eq:chaining}
e^{k-1}/k^k+
7(3/\size\Sigma)^{d+1}+1/2^{\size\Sigma/2}.
\end{equation}
If $k$ is not too large, the first
term dominates.

\medskip\noindent
\textbf{Lower bound.} Finally, we show that our upper bound for the error probability in~\Cref{thm:tech-random-set} is tight within a constant factor.
\begin{restatable}{theorem}{lowerbound}\label{thm:lower-bound} Let $h=\widehat h\circ \derive h:\Sigma^c\to\cR$
	be a random tornado tabulation hash function with $d$ derived characters. There exists a selector function $f$ with $\mu^f \leq  \Sigma / 2$ such that the 
	derived selected keys $\derive h(X^{f,h})$ are linearly
	dependent with probability
	at least $\Omega((3/|\Sigma|)^{d-2})$.
\end{restatable}

\medskip\noindent
\textbf{Full randomness for larger sets of selected keys.} As mentioned earlier, for
a fast implementation of tabulation hashing, we pick the alphabet
$\Sigma$ small enough for the tables to fit in fast cache.
A common choice is 8-bit characters.
However, we only get
full randomness for (selected) sets of (expected) size at most $\size\Sigma/2$ (c.f. Theorems
\ref{thm:intro-random-set} and \ref{thm:intro-fixed-set}).

To handle larger sets, we
prove that it suffices to 
only increase the alphabet of
the last two  derived characters;
meaning that only two lookups
have to use larger tables.
This is close to best possible
in that we trivially need at least
one lookup in a table at least as
big as the set we want full
randomness over. More precisely,
if we use alphabet $\Psi$ for
the last two derived characters,
then our size bound increases
to $\size\Psi/2$. The
error probability bound of 
Theorems \ref{thm:intro-random-set}
becomes	
$$14 (\mu^f)^3(3/|\Psi|)^2(3/|\Sigma|)^{d-1} + 1/2^{\size{\Sigma}/2}\;.$$
With the above mix of alphabets,
we have tornado hashing running
in fast cache except for the last
two lookups that could dominate
the overall cost, both in time
and space.  Because they dominate,
we will consider a slight variant,
where we do not store derived
characters in any of the two large tables. Essentially this means
that we change the definition
of the last derived character
from $\derive x_{c+d}=\derive h_d(\derive x_1\cdots\derive x_{c+d-1})$ to $\derive x_{c+d}=\derive h_d(\derive x_1\cdots\derive x_{c+d-2})$. This is going to cost us
a factor two in the error probability, but in our implementation, we will now have $c+d-2$ lookups in tables
$\Sigma\to \Sigma^{d-1}\times\Psi^2\times R$ (where the values are represented as a single $w$-bit numbers), and
2 lookups in tables $\Psi\to  R$.
We shall refer to this scheme
as \emph{tornado-mix}. Corresponding to \Cref{thm:intro-random-set}, we get

\begin{restatable}{theorem}{tornadomixrandomtheorem}\label{thm:tornadomixrandom}
	Let $h=\widehat h\circ \derive h:\Sigma^c\to\cR$
	be a random tornado-mix tabulation hash function with $d$ derived characters, the last two from $\Psi$, and an $s$-selector function $f$. If $\mu=\mu^f \leq \size\Psi/ 2$
	then $h^{(t)}$ is fully random on $X^{f,h}$ with probability at least
	$$1 - 14 \mu^3(3/|\Psi|)^2(3/|\Sigma|)^{d-1} - 1/2^{\size{\Sigma}/2}\;.$$
\end{restatable}

An interesting case
is when we want linear space
uniform hashing for a fixed set $X$, in which case $\mu=\size X$ above. This
leads to a much better implementation of the uniform
hashing fo Pagh and Pagh \cite{PP08}. The main part
of their paper was to show a linear
space implementation, and
this would suffice for all but
the most succinct space algorithms.
They used highly independent
hashing \cite{siegel04hash,thorup13doubletab} as a subroutine, but this
subroutine alone is orders of magnitude slower than our simple
implementation (see, e.g., \cite{Aamand0KKRT22}). They 
combined this with a general reduction from
succinct space to linear space, for which we now have a really efficient construction.

\subsection{Techniques and relation to mixed tabulation}\label{sec:intro-related-work}

In spirit, our results are very related to the results on mixed tabulation \cite{dahlgaard15k-partitions}. For now, we only
consider the case of a single
alphabet $\Sigma$.
Indeed, tornado and mixed tabulation are very similar to implement. Both deal with $c$-character keys from
some alphabet $\Sigma$, produce
a derived key with $c+d$ characters, and then
apply a top simple tabulation to the
resulting derived keys.  Both
schemes can be implemented with
$c+d$ lookups. The difference is in how the two schemes compute the derived keys.
For ease of presentation, let
$\derive h_i:\Sigma^*\to\Sigma$, that
is, $\derive h_i$ adjusts to the
number of characters in the input.
Now for mixed tabulation,
we define the derived key
$\derive x_1\cdots \derive x_{c+d}$
by
\begin{equation*}
	\derive x_i = \begin{cases}
		x_i &\text{if $i\leq c$} \\
		\derive h_{i-c}\ld(\derive x_1 \dots \derive x_{c}\rd) &\text{if $i > c$}.
	\end{cases}
\end{equation*}
The analysis from \cite{dahlgaard15k-partitions} 
did not consider query keys, but
ignoring this issue, their
analysis works in the limit $\size\Sigma\to\infty$. For 
example, the analysis from \cite{dahlgaard15k-partitions} requires that \footnote{While using their Lemma 3, if $t$ goes up to $s$ in case D.} 
\[(\log|\Sigma|)^c\leq |\Sigma|/2.\]
This is true for
$c=O(1)$ and $|\Sigma|\to\infty$,
but simply false for 
realistic parameters. 
Assuming the above condition to be satisfied,
if we consider scenarios with non-constant $c$ and $d$, the error probability from \cite{dahlgaard15k-partitions} becomes
\[\left(O(cd)^c/
|\Sigma|\right)^{\lfloor d/2\rfloor-1}+1/2^{\Omega(|\Sigma|)}.\]
Now, even if we replace $O(cd)$
with $cd$, the error probability
is not below 1 even with $16$-bit
characters and $c= 4$. In contrast, practical tabulation schemes normally
use $8$-bit characters for efficiency, and our explicit 
bound of $\DP$ from 
\Cref{thm:tech-random-set} works fine even in this case, implying that
our theory actually applies to practice.

The reason that mixed tabulation
has the problematic exponential dependency on $c$ is that for a
set of linearly dependent keys, 
it uses a clever encoding
of each of the $c$ characters in
some of the keys. With tornado, the only encoding
we use is that if we have a zero set of keys, then each key is
the xor of the other keys in
the zero set, and this
is independent of
$c$.~\footnote{A set is linearly dependent if it contains a subset that is a zero set. See~\Cref{sec:prelim} for the precise definition.}

To describe  the advantage
of tornado tabulation over
mixed tabulation, it is easier
to first compare with what
we call \emph{simple tornado} 
where the derived key $\derive x_1\cdots \derive x_{c+d}$
is defined by
\begin{equation*}
	\derive x_i = \begin{cases}
		x_i &\text{if $i\leq c$} \\
		\derive h_{i-c}\ld(\derive x_1 \dots \derive x_{i- 1}\rd) &\text{if $i > c$}.
	\end{cases}
\end{equation*}
The implementation difference is that with
simple tornado, each derived
character uses all preceding characters
while mixed tabulation
uses only the original characters. This difference gives tornado
a big advantage when it comes to
breaking up linear dependence in
the input keys.
Recall that a set $Y \subset \Sigma^{c+d}$ of derived keys is linearly independent if and only if, for every subset $Y'\subseteq Y$, there exists a character position $i\in\set{1,\ldots, c+d}$ such that some character appears an odd number of times in position $i$ among the keys in $Y'$. Intuitively, the strategy is to argue that whatever linear dependence exists in the input key set to begin with (essentially, in the first $c$ characters of the derived keys) will,  whp, be broken down by their $d$ derived characters (and hence, disappear in the derived keys).
In this context, the way we compute the derived characters becomes crucial:  in mixed tabulation,  each derived character is computed independently of the other derived characters. Thus, whether a derived character breaks a dependency or not is independent of what other derived characters do. In contrast, we make the derived characters in tornado tabulation depend on all previously computed derived characters such that,  if we know that some derived character does not break a dependency, then we also know that none of its previously computed derived characters have broken it either. Or, in other words, each successive tornado-derived character benefits from the dependencies already broken by previously computed derived characters, i.e., the benefits compound each time we compute a new derived character. 

This structural dependence between tornado-derived characters turns out to be very powerful in breaking linear dependencies among the input keys and indeed, leads to a much cleaner analysis.
The most important benefit, however, is that  tornado tabulation has a much lower
error probability. For
simple tornado, we get an 
error probability of
\begin{equation*}
	7 (\mu^f)^3(3/|\Sigma|)^{d} + 1/2^{\size{\Sigma}/2}\;\textnormal,
\end{equation*}
which essentially gains a factor
$(3/|\Sigma|)$ for each derived
character. It turns out that
we gain an extra factor $(3/|\Sigma|)$ if
we twist the last original 
character as we did in the original tornado definition \req{eq:tornado},
and then we get the bound from \Cref{thm:intro-random-set}, the point being that this twisting
does not increase the number of lookups.

One might wonder if twisting
more characters would help
further, that is, setting
$\derive x_i=x_i\xor\derive h_i(\derive x_1\cdots x_{i-1})$
for $i=2,\ldots,c$, but it doesn't.
The point is that the bad key set from our lower bound
in \Cref{thm:lower-bound} is of
the form $[0]^{c-2}\times[2]\times A$ for some $A \subseteq \Sigma$, and then it is only the last
character where twisting makes a difference.

As a last point, recall tornado-mix which was designed to deal with
larger sets. There it 
only costs us a factor
2 when we let the last
derived character $\derive x_{c+d}$ depend only on 
$\derive x_1 \dots \derive x_{c+d-2}$
and not on $\derive x_{c+d-1}$. This
is essentially like switching to
mixed tabulation on the last
two derived keys, hence the name \emph{tornado-mix}. This only works for the last two derived characters that can play a symmetric role.

\medskip\noindent
\textbf{The exponential dependence on $c$.}
We note that having an exponential dependence
on $c$ is symptomatic for almost
all prior work on tabulation hashing, starting from the original work in \cite{patrascu12charhash}. 
Above, we discussed how tornado tabulation
hashing avoided such exponential dependence on $c$
in the error probability from
mixed tabulation. The dependence
on $c$ was particularly destructive for mixed tabulation because it pushed the 
error probability above 1 for
the relevant parameters.

\medskip\noindent
\textbf{Concentration bounds.}
Tornado hashing inherits the
strongest concentration bounds
known for mixed tabulation~\cite{HouenT22:chaos}.
The reason is that they only require one derived character
and tornado tabulation can 
be seen as applying mixed tabulation with one character
to a derived tornado key with one less character. Unlike
our Lemma \ref{lemma:upper-chernoff}, which essentially only applies
for expected values $\mu\leq\size\Sigma/2$, the  concentration
bounds we get from~\cite{HouenT22:chaos}
work for unbounded $\mu$ and for both the upper and lower tail.
They fall exponentially in $\mu/K_c$
where $K_c$ is exponential in $c$,
but this still yields strong
bounds when $\mu$ is large. Inheriting the strongest concentration bound for mixed tabulation implies that
tornado tabulation can replace mixed
tabulation in all the applications from \cite{dahlgaard15k-partitions} while improving on the issue of
local uniformity.

\subsection{The power of locally uniform hashing}\label{sec:intro-local-power}
We now describe, on a high level, how the notion of locally uniform hashing captures a type of randomness that is sufficient for many applications to work almost as if they employed full randomness. For the discussion, we will assume the parameters from Theorem \ref{thm:tornadomixrandom} with tornado-mix.
\drop{Consider some algorithm and a performance measure that is monotonically increasing in the number of input keys.} 
The main observation is that, for many algorithms, the performance measure (i.e., its distribution) depends, whp, only on the behavior of keys that hash inside a local neighborhood defined via some selected bits of the hash value
(the selection may depend both on the algorithm and on the concrete input since the selection is only used for analysis). Moreover, the set of keys  $X^{f,h}$ that land in that selected neighborhood has
expected size $\leq \size\Psi/2$. For any such neighborhood, our results imply that the keys in $X^{f,h}$ have fully
random free bits whp.

To understand the role of the free bits, it is helpful to think of hashing a key $x$ as a two-stage process: the selector bits of $h(x)$ tell us whether $x$ hashes in the desired neighborhood or not, while the remaining free bits of $h(x)$ determine how $x$ hashes once it is \emph{inside} the neighborhood. This suggests a general coupling, where we let both fully-random hashing and tornado tabulation hashing first choose the select bits and then the free bits. Since both are fully random on the free bits (for us with high probability), the only difference is in the selection, but here concentration bounds imply that we select almost the same number of keys as fully-random hashing. 
\drop{Our results then imply that any local performance measure that depends only on the keys in the neighborhood is stochastically dominated, when using tornado tabulation hashing, by its counterpart when using fully-random hashing on slightly more input keys. Thus, any corresponding analysis that assumes fully-random hashing can be employed in a black-box fashion for tornado tabulation hashing.}

In \cite{dahlgaard15k-partitions}, this approach was demonstrated
for the  HyperLogLog algorithm of 
Flajolet, Fusy, Gandouet,
Meunier~\cite{Flajolet07hyperloglog}
for counting distinct elements in data streams, and the One-Permutation Hashing of Li, Owen, and Zhang ~\cite{li12oneperm} used
for fast set similarity in large
scale machine learning. In
\cite{dahlgaard15k-partitions} they
implemented the local uniformity
with mixed tabulation, but with
tornado hashing we get a realistic
implementation. This also includes later application of mixed tabulation, e.g., the dynamic load balancing from \cite{AamandKT21:dynamic-load}.
Below, as a new example,  we illustrate how local uniformity makes classic linear probing perform almost
as if fully random hashing was used.

We briefly recall the basic version of linear probing. We employ an array T of length $m$ and hash keys to array entries using a tornado tabulation hash function $h:\Sigma^c\rightarrow [m]$. Thus, in Theorem \ref{thm:tornadomixrandom}, we have
$\cR=[m]=[2^r]$.
Upon insertion of a key $x$, we first check if entry $T[h(q)]$ is empty. If it is, we insert the key in $T[h(q)]$. Otherwise, we scan positions in $T$ sequentially starting with $T[h(q)+1]$ until we find an empty position and store $x$ in this next available free entry. To search for $x$, we inspect array entries sequentially starting with $T[h(q)]$ until we find $x$ or we find an empty array entry, concluding that the key is not in the array. The general concept that dominates the performance of linear probing is the \emph{run} of $q$, which is defined as the longest interval $I_q \subseteq [m]$ of full (consecutive) array entries that $h(q)$ is part of. The run affects the probe length (the length of
the interval from $h(q)$ till the
end of $I_q$) and other monotone measures such as insertion and deletion time  (i.e., monotonically increasing in the number of keys). Thus, any corresponding analysis depends only on the set $X_q$ of keys that fall in the interval $I_q$, and it is there that we are interested in showing full randomness.

The first step is to argue that the behavior of $I_q$ is 
local in nature, in that it is affected only by the keys that hash in a fixed length interval around $h(q)$. To that end, we show that, whp, the length of the run is no bigger than a specific $\Delta=2^\ell$. This $\Delta$ implies locality: consider the interval $J_q\subseteq [m]$, centered at $h(q)$, which extends $\Delta$ entries in each direction, i.e., $J_q = \set{ j\in[m] \mid \size{ j - h(q)} \leq \Delta}$. Then, whp, any run that starts outside $J_q$ is guaranteed to end by the time we start the run of $q$.  In other words, whp, the behavior of $I_q$ depends only on the set $X_q$ of keys that hash inside the fixed length neighborhood $J_q$ (and specifically where in $J_q$ the keys hash). A similar argument can be made for other variants of linear probing, such as linear probing with tombstones~\cite{DBLP:conf/focs/BenderKK21}, where the run/neighborhood must also take into account the tombstones that have been inserted.

At this point, it is worthwhile pointing out that the set $X_q$ is no longer a fixed set of keys but rather a random variable that depends on the realization of $h(q)$. We cannot just apply~\Cref{thm:intro-fixed-set}. Nevertheless, in the second step, we argue that $X_q$ can be captured by our notion of selector functions. For this, we cover $J_q$ with three dyadic intervals, each of length $\Delta$: one including $h(q)$ and the corresponding ones to the left and to the right. Since each interval is dyadic, it is determined only by the leftmost $r-\ell$ bits of the hash value: for example, an element $x\in\Sigma^c $ hashes into the same dyadic interval as $q$ iff the leftmost $r-\ell$ bits of $h(x)$ match those of $h(q)$. We can then design a selector function that returns a $1$ if and only if $x$ is in the input set and the leftmost $r-\ell$ bits of $h(x)$ (its selector bits) hash it into any of the three specific dyadic intervals we care about. Such a selector function is guaranteed to select all the keys in $X_q$. By setting $\Delta$ appropriately, we get that the expected size of the selected set is at most $\Sigma/2$, and can thus apply~\Cref{thm:intro-random-set}. We get that inside the intervals, keys from $X^{f,h}$ hash (based on their free bits) in a fully random fashion.

Finally, what is left to argue is that the number of keys hashing inside each of the three dyadic intervals is not much bigger than what we would get if we used a fully-random hash function for the entire set (one in which the selector bits are also fully random).  For this, we employ~\Cref{lemma:upper-chernoff}, and can conclude that we perform almost as well with fully-random hashing.  In particular, from Knuth's \cite{knuth63linprobe} analysis of linear
probing with fully-random hashing, we conclude that with load $(1-\eps)$, the expected probe length
with tornado hashing is $(1+o(1))(1+1/\eps^2)/2$.

The above type of analysis could be applied to other variants of linear probing, e.g., including lazy deletions and tombstones as in the recent work of 	Bender, Kuszmaul, and Kuszmaul \cite{DBLP:conf/focs/BenderKK21}. We would need to argue
that the run, including tombstones, remains within the selected neighborhood and that we have concentration both on live keys and tombstones. However, since the analysis from \cite{DBLP:conf/focs/BenderKK21} is already based on simple hash functions, we would not get new bounds from knowing that tornado hashing performs almost as well as fully-random hashing.

\subsection{Relation to high independence}\label{sec:intro-high}
The $k$-independence approach to hashing \cite{wegman81kwise} is to construct
hash functions that map every set of $k$ keys independently and uniformly at random. 
This is much stronger than getting
fully random hashing for a 
given set and, not surprisingly,
the results for highly independent
hashing \cite{christiani15indep,siegel04hash,thorup13doubletab} are weaker. The strongest high independence result from \cite{christiani15indep} says that we get $|\Sigma|^{1-\eps}$-independent hashing in $O((c/\eps)\log(c/\eps))$ time,
but the independence is much less
than our $|\Sigma|/2$.  Experiments
from \cite{Aamand0KKRT22} found
that even the simpler non-recursive highly independent hashing from \cite{thorup13doubletab}
was more than an order of magnitude
slower than mixed tabulation which is
similar to our tornado tabulation.

We also note that the whole idea of using derived
characters to get linear independence between
keys came from attempts to get higher  $k$-independence 
Indeed, \cite{dietzfel03tabhash,KW12, thorup12kwise} derive
their characters deterministically, guaranteeing
that we get $k$-independence. The construction
for $k=5$ is efficient but for larger $k$,
the best deterministic construction is that
from \cite{thorup12kwise} using 
$d=(k-1)(c-1) + 1$ derived characters,
which for larger $k$ is much worse than the
randomized constructions.

\subsection{Relation to the splitting trick and to succinct uniform hashing} \label{sec:splitnshare}
We will now describe how tornado-mix provides a much more efficient
implementation of the succinct
uniform hashing by Pagh and Pagh \cite{PP08} and the splitting
trick of Dietzfelbinger and Rink \cite{dietzfel09splitting}. This further illustrates how our work provides a component of wide applicability within hashing.

\medskip\noindent
\textbf{Succinct uniform hashing.}
The main contribution in \cite{PP08} is to obtain uniform hashing in linear space. The succinct construction is then obtained through a general reduction from the linear-space case. They achieved uniform hashing in linear space using the highly independent hashing of Siegel \cite{siegel04hash} as a subroutine in combination with other ideas. Now, thanks to \Cref{thm:tornadomixrandom}, we know that tornado-mix offers an extremely simple and efficient
alternative to implement that. Indeed, if $n$ is the size of the set we want linear space uniform hashing on, then setting $|\Psi|$
to the power of two just above $2n$ is sufficient to ensure the degree of independence that we need. Moreover, we can set $\size\Sigma\sim\sqrt\Psi$ and $c$ so that we obtain (i) $c |\Sigma| = o(|\Psi|)$ and (ii) setting $d = 2c+1$ the probability bound in \Cref{thm:tornadomixrandom} becomes $1 - O(1/u)$ where $u$ is the size of the universe. The two previous conditions ensure respectively that tornado-mix uses linear space and that it is, whp, fully random. 

\medskip\noindent
\textbf{The splitting trick.} 
We now explain how our tornado-mix tabulation hashing provides a very simple and efficient implementation of the 
splitting trick of Dietzfelbinger and Rink \cite{dietzfel09splitting} which is a popular method for simulating full-randomness in the context of data structures, most notably various dictionary constructions~\cite{DIETZFELBINGER200747,10.1145/780542.780634, DBLP:journals/mst/FotakisPSS05,DBLP:conf/focs/ArbitmanNS10}. This is an especially relevant application because it usually requires hashing keys into a range of size $ \Theta(n)$ and, in this context, employing the uniform hashing of Pagh and Pagh \cite{PP08} would be too costly, i.e., (compared to the size of the overall data structure). This trick was also used to obtain a simpler and exponentially faster implementation of succinct uniform hashing \cite{PP08}.
We note that the splitting idea had also been used earlier works of Dietzfelbinger and Meyer auf der Heide \cite{DietzfelbingerH90} and Dietzfelbinger and Woelfel \cite{dietzfel03tabhash}.

The idea is to first split the input set $S$ of size $n$ into $n^{1-\delta}$ subsets $S_1,\ldots, S_m$, for some $\delta \in (0,1)$. The splitting is done through a hash function $h_1:\Sigma^c\rightarrow [n^{1-\delta}]$ such that $S_i$ is defined as all the keys in $S$ that hash to the same value, i.e., $S_i = \set{x\in S \mid h_1(x)=i}$. 
In many applications, we
want the splitting to be balanced, that is, we need a joint upper bound $s$ on
all set sizes with $s$ close to
the expected size $n^\delta$.
On top of this, we need a shared
hash function $h_2:\Sigma^c\to \cR$
which with high probability is
fully random on each set $S_i$ and
we ensure this with a hash function
that w.h.p. is fully random on any set of size at most $s$. The problem is then solved separately for each subset (e.g., building $n^{1-\delta}$ dictionaries, each responsible for just one subset $S_i$).

The above splitting trick can be easily done with tornado-mix hashing from Theorem \ref{thm:tornadomixrandom}. 
The dominant cost for larger $s$ is two lookups in tables of size at most $4s$. We let $h_1$ be
the select bits (with the strong concentration from \Cref{lemma:upper-chernoff}) and $h_2$ as the remaining free bits. Getting both for the price of one is nice, but the most important thing is that we get an efficient implementation of the shared hash function $h_2$. Specifically,
we compare this with how $h_1$ and $h_2$ were implemented in \cite{dietzfel09splitting} to get uniform hashing. For the splitter $h_1$, instead
of using Siegel's \cite{siegel04hash} highly independent hashing,   \cite[\S 3.1]{dietzfel09splitting} uses that if $h_1$ has sufficiently high, but constant, independence, then it offers good concentration (though not as good as ours from \Cref{lemma:upper-chernoff}).

The bigger issue is in the implementation of the shared $h_2$ from
\cite[\S 3.2]{dietzfel09splitting}.
For this, they first
employ a hash function $f:[s]\to [s^{1+\eps}]$ such that, whp, $f$ has  at most $k=O(1)$ 
keys that get the same hash value. For
small $\eps$ and $k$, this forces
a high independence of $f$. Next,
for every $i\in [s^{1+\eps}]$, they
use a $k$-independent hash
function $g_i:\Sigma^c\to\cR$, and finally, $h_2$ is implemented as $g_{f(x)}(x)$.
The space to implement $h_2$ is dominated by the $O(s^{1+\eps})$ space to store
the $s^{1+\eps}$ $k$-independent
hash functions $g_i$, and time-wise, we have to run several sufficiently independent hash functions. This should be compared with our tornado-mix that 
only uses $O(s)$ space and runs in
time corresponding to a few multiplications (in tests with 32-bit keys). 

\subsection{Paper Organization}
The remainder of our paper is structured as follows. In~\Cref{sec:prelim}, we introduce some notation, a more general notion of generalized keys and~\Cref{lem:simple-tab-on-lin-indep-sets} (the corresponding version of~\Cref{lem:simple-tab-on-lin-indep-sets-regular} for them). Our main technical result, ~\Cref{thm:tech-random-set}, is proved in three steps: we first define obstructions, the main combinatorial objects we study, in~\Cref{sec:obstructions}. In \Cref{sec:simplified} we present a simplified analysis of~\Cref{thm:tech-random-set} that achieves a weaker error probability. We then show in~\Cref{sec:tighter} a tighter analysis that finally achieves the desired bound. The Chernoff bounds for the upper tail are proved in~\Cref{sec:upper-tail-chernoff}. The details for the linear probing analysis can be found in~\Cref{sec:linear-probing}. Finally, the lower bound from~\Cref{thm:lower-bound} is proved in~\Cref{sec:lower-bound}.

\section{Preliminaries}\label{sec:prelim}
\paragraph{Notation.} We use the notation $n!!=n(n-2)\cdots$. More precisely, $n!!=1$ for $n\in \{0,1\}$,
while $n!!=n(n-2)!!$ for $n>1$. For odd $n$, this is exactly the number
of perfect matchings of $n+1$ nodes. We use the
notations
\begin{align*}
	n^{\ul k}&=n(n-1)\cdots (n+1-k)=n!/(n-k)!\\
	n^{\ull k}&=n(n-2)\cdots (n+2(1-k))=n!!/(n-2k)!
\end{align*}
Here $n^{\underline{\underline k}}$ appears to be non-standard, though
it will be very useful in this paper in connection with
something we will call greedy matchings. We note that
\[n^{\underline{\underline k}}\leq n^{\underline k}\leq n^k.\]

\paragraph{Position characters and  generalized keys.} We employ a simple generalization of keys going back to P{\v a}tra{\c s}cu  and Thorup~\cite{patrascu12charhash}. Namely, a \emph{position character} is an element of $\{1 \dots b\} \times \Sigma$, e.g., where $b=c$ or $c+d$. Under this definition a key $x \in \Sigma^b$ can be viewed as a set of $b$ position characters $(1, x_1) \dots (b, x_b)$, and, in general, we consider \emph{generalized keys} that may be arbitrary subsets of $\{1 \dots b\} \times \Sigma$. A natural example of a generalized key is the symmetric difference
$x\sd y$ of two (regular) keys. We
then have that
\[h(x)=h(y)\iff h(x\sd y)=0.\]
These symmetric differences will play an important role in our
constructions, and we shall refer to $x\sd y$ as a \emph{diff-key}.  A
diff-key is thus a generalized key where we have zero or two characters
in each position. For a generalized key $x$, we can then  define
\[x[i]=\{(i,a) \in x\}\textnormal,\quad  x[<i]=\{(j,a) \in x\mid j<i\}\quad\textnormal{and}\quad 
x[\leq i]=\{(j,a) \in x\mid j\leq i\}.\]
For example, if we have a regular key $x=x_1,\ldots,x_c$, then $x[i]=x_i$
and $x[\leq i]=x_1\ldots,x_i$. Also note that 
$(x\sd y)[\leq i]=
x[\leq i]\sd y[\leq i]$.
We shall also apply this indexing
to sets $X$ of generalized keys by applying it to each key individually, e.g.,
\[X[<i]=\{x[<i]\mid x\in X\}.\]

\paragraph{Generalized keys and linear independence.} A generalized key $x$ can also be interpreted as a $(|\Sigma| \cdot b)$-dimensional vector over $\bF_2$, where the only entries set to $1$ are those indexed by position characters in $x$. The generalized key domain is
denoted by $\bF_2^{\{1 \dots b\} \times \Sigma}$.  Now, if we have a simple tabulation hash function $h:\Sigma^b\to \cR$ using character tables $T_1,\ldots,T_b$, then $h$ can be lifted to hash any generalized key  $x\in \bF_2^{\{1 \dots b\} \times \Sigma}$ by
\[h(x)=\bigoplus_{(i,a)\in x} T_i[a].\]
Thus $h$ provides a mapping
$\bF_2^{\{1 \dots b\} \times \Sigma}$ to $\cR$.

As for (regular) keys, for a  set $Y$ of generalized keys, we can then define
$\bigsd Y$ to be the set of position characters that appear an odd
number of times across the keys in $Y$.
We then say that $Y$ is a \emph{zero-set} if $\bigsd Y=
\emptyset$. We also say that $Y$ is \emph{linearly dependent} if it contains a
subset which is a zero-set, and \emph{linearly independent} otherwise.
It is apparent then that a set of generalized keys is linearly independent if and only if the set of their vector representations is linearly independent. Indeed, the proof of Thorup and Zhang  \cite{thorup12kwise} for Lemma \ref{lem:simple-tab-on-lin-indep-sets} works quite directly in this generality. Thus we have
\begin{lemma}[Simple tabulation on linearly independent generalised keys] \label{lem:simple-tab-on-lin-indep-sets}
	Given a set of generalized keys $Y\subseteq\left(\bF_2^{\{1 \dots c\} \times \Sigma}\right)$ and a simple tabulation hash function $h: \Sigma^c \rightarrow \cR$, the following are equivalent:
	\begin{enumerate}[label = (\roman*)]
		\item $Y$ is linearly independent
		\item $h$ is fully random on $Y$ (i.e., $h|_Y$ is distributed uniformly over $\cR^{Y}$).
	\end{enumerate}
\end{lemma}

\section{Obstructions with simple tornado tabulation}\label{sec:obstructions}

We prove~\Cref{thm:tech-random-set} by first considering a simpler version  of tornado tabulation hashing, which we call \emph{simple tornado hashing}, where we do not
change the last  character
of the (original) key. Formally,  for
a key $x=x_1\cdots x_c$, its corresponding derived key $\derive x = \derive x_1 \ldots \derive x_{c+d}$ is computed as
\begin{equation*}
	\derive x_i = \begin{cases}
		x_i &\text{if $i = 1, \ldots ,c$} \\
		\derive h_{i-c}\ld(\derive x_1 \dots \derive x_{i-1}\rd) &\text{otherwise}.
	\end{cases}
\end{equation*}
Note that, in the original tornado hashing,
we had $\derive x_c=x_{c} \xor \derive h_{0}(\derive x_1 \dots \derive x_{c-1})$. Removing this extra step is thus equivalent to fixing $\derive h_0(\cdot)=0$. While this step comes  at almost no cost in the code, it allows us to gain a factor of $3/\size{\Sigma}$ in the overall error probability. See~\Cref{sec:add-twist} for details.  For the simple tornado hashing, we will
prove a slightly weaker probability bound.

For ease of notation,  for every key $x$, we use $\derive x$ to denote
the corresponding derived key $\derive h(x)$; and likewise for any set
of keys. We also define $X=X^{f,h}$ to be the set of selected keys and
$\derive X=\tilde h(X)$. 
We want to argue that the derived selected keys $\derive X$ are linearly independent with high
probability. To prove this, we assume that $\derive X$ is linearly dependent and hence, contains some zero-set $\derive Z$. From $\derive Z$ we
construct  a certain type of ``obstruction''  that we show is unlikely to occur.

\subsection{Levels and matchings}
We first define some necessary concepts.  We use the notion of \emph{level} $i$ to refer to position $c+i$ in the derived keys. 
Let $M \subseteq \binom{\ssigma^c}{2}$ be a (partial) matching
on the keys $\Sigma^c$.\footnote{Here, we mean the graph-theoretic definition of a matching as a set of edges with disjoint endpoints. In our case, the vertices of the graph are keys in $\Sigma^c$, and the edges of the matching are represented as $\set{x,y} \in M$. }
We say that $M$ is an \emph{$i$-matching} if for all $\{x,y\}\in M$, 
it holds that $\derive x[c+i]=\derive y[c+i]$, namely if every pair of keys in $M$ matches on level $i$.
Our obstruction will, among other things, contain an $i$-matching $M_i$
for each level $i$.

Recall that a diff-key  $x \triangle y$ is the symmetric difference of two keys $x$ and $y$ in terms of their position characters.
We then say that $M$ is an \emph{$i$-zero, $i$-dependent, or $i$-independent} matching
if 
\[\DiffKeys(M,i)=\left\{(\derive x\sd\; \derive y)[\leq c+i]\mid \{x,y\}\in 
M\right\}\] 
is a zero-set, linearly dependent, or linearly independent, respectively. In other words, for each pair $\set{x,y}$ in the matching, we consider the diff-key
corresponding to the first $c+i$ characters of their derived keys. We then ask if this set of diff-keys now forms a zero-set or contains one as a subset,  by looking at their (collective) symmetric difference. We employ this notion to derive the probability that the function $\derive h$ satisfies a certain matching as such:

\begin{lemma}\label{lem:independent-matchings}
	Let $M$ be a partial matching on $\Sigma^c$. Conditioning on $M$ being $(i-1)$-independent, $M$ is an $i$-matching 
	with probability $1/|\Sigma|^{|M|}$.
\end{lemma}
\begin{proof} 
	First, we notice that the event ``$M$ is $(i-1)$-independent'' only depends on $\derive h_j$ for $j<i$. Then, by 
	Lemma \ref{lem:simple-tab-on-lin-indep-sets}, when we apply the
	simple tabulation hash function $\derive h_i:\Sigma^{c+i}\to\Sigma$
	to the linearly independent generalized key set $\DiffKeys(M,i-1)$,
	the resulting hash values $\derive h(z)[c+i]$ for
	$z\in\DiffKeys(M,i-1)$ are independent and uniform over
	$\Sigma$. Hence, so are the resulting derived characters
	$\derive h(z)[c+i]$ for
	$z\in\DiffKeys(M,i)$. The probability that they are all $0$ is
	therefore $1/|\Sigma|^{|M|}$.
\end{proof}

Similarly, as for matchings, we say that a set of keys $Z\subseteq\Sigma^c$ is  \emph{$i$-zero}, \emph{$i$-dependent, or $i$-independent} if $\derive Z[\leq c+i]$ is a zero-set, linearly dependent, or linearly independent, respectively. We note the following relations:

\begin{observation}\label{lem:perfect-matchings-to-zero-set}
	Let $M$ be a partial matching on $\Sigma^c$ and $Z=\bigcup M$. 
	Then $M$ is an $i$-zero matching iff $Z$ is an $i$-zero set.
    Furthermore, if $M$ is $i$-dependent then $Z$ is also $i$-dependent (but not vice versa). 
\end{observation}
\medskip\noindent
We will also make use of the following observation repeatedly:
\begin{observation}\label{lem:zero-set-to-perfect-matchings}
	If $Z$ is an $i$-zero set, then there is a perfect $j$-matching on $Z$ for every 
	level $j\leq i$.
\end{observation}

\medskip\noindent

\subsection{Constructing an obstruction}\label{sec:construct}
In this section, we show that whenever the set of selected derived keys $\derive X$ is linearly dependent, it gives rise
to a certain \emph{obstruction}. We now show how to construct such an obstruction.

Since $\derive X$ is linearly dependent, there
must be some $d$-zero set $Z \subseteq X$. We
are going to pick a $d$-zero set that
(1) minimizes the number of elements contained that are not in the query set $Q$ and, subject to (1), (2) minimizes the number of elements from $Q$ contained. In particular, $Z$ is contained in $Q$ if $\derive Q$ is not linearly independent.

If $Z$ is not contained in $Q$, we let $x^*$ be any element from $Z\setminus Q$. Then
$Q \cup Z \setminus \{x^*\}$ must be linearly independent since any strict subset $Z'$ would
contradict (1). If $Z$ is contained in $Q$, then we let $x^*$ be
an arbitrary element of $Z$.

\paragraph{The top two levels.}
By~\Cref{lem:zero-set-to-perfect-matchings}, we have a
perfect $d$-matching $M^*_{d}$ and a perfect $(d-1)$-matching
$M^*_{d-1}$ on $Z$ (we also have perfect matchings on other levels, but they
will be treated later). These two perfect matchings partition $Z$
into alternating cycles.

We will now traverse these cycles in an
arbitrary order except that we traverse the
cycle containing $x^*$ last. For all
but the last cycle, we start in an arbitrary vertex and its cycle starting with the
edge from $M^*_{d-1}$. When we get to
the last cycle, we start at the $M^*_{d}$ neighbor of $x^*$ and follow the cycle from the $M^*_{d-1}$ neighbor of this key. This
ensures that $x^*$ will be the very last
vertex visited. The result is a traversal sequence $x_1,\ldots,x_{|Z|}$
of the vertices in $Z$ ending in $x_{|Z|}=x^*$. Note that $M^*_{d-1}$ contains the pairs
$\{x_{1},x_2\},\{x_3,x_4\},\ldots$. For $M^*_d$ it is a bit more complicated,
since its pairs may be used to complete a cycle (and hence are not visible in the traversal).

We now  define $W= \{x_1 \dots x_w\}$ to be the shortest prefix of
$x_1 \dots x_{|Z|}$ such that $M^*_{d-1}$ restricted to the keys in $W$ is a
$(d-1)$-dependent matching, i.e., we go through the pairs $\set{x_1,x_2}, \set{x_3,x_4}, \ldots$ until the set of their diff-keys (up to level $d-1$) contains a zero-set. Note that such a $W\subseteq Z$ always exists because $M^*_{d-1}$ itself is a $(d-1)$-zero matching. Also note
that $W\setminus\{x_w\}\subseteq Z\setminus\{x^*\}$.
Let $e_{d-1} := \{x_{w-1}, x_w\}$ be the last pair in the prefix, and as $e_{d-1} \in
M^*_{d-1}$, we get that $w$ is even.
We then define $M_{d-1}$ to be the
restriction of $M^*_{d-1}$ to the  keys in $W$ and $M_d$ to be the
restriction of $M^*_d$ to the keys in $W\setminus \set{x_w}$. Note that
$M_{d-1}$ is a perfect matching on $W$ while $M_d$ is a maximal
matching on $W\setminus \{x_w\}$.  Since $M_{d-1}$ is
$(d-1)$-dependent, we can define a submatching $L_{d-1} \subseteq
M_{d-1}$ such that $L_{d-1}$ is a $(d-1)$-zero matching (this corresponds exactly to the subset of $\DiffKeys(M_{d-1},d-1)$ that is a zero-set). By
construction, $e_{d-1} \in L_{d-1}$. Finally, we set $Z_{d-1}=\bigcup
L_{d-1}$ and notice that $Z_{d-1}$ is an $(d-1)$-zero key set (by~\Cref{lem:perfect-matchings-to-zero-set}).

\paragraph{A special total order.}
We now define a special new total order $\preceq$ on $\Sigma^c$ that we use in order to index the keys in $W$ and describe matchings on levels $<d-1$.  Here $x_w$ has a special role
and we place it in a special position; namely at the end. More precisely, we have the natural $\leq$-order  on $\Sigma^c$, i.e, in which keys are viewed as numbers $<|\Sigma|^c$. 
We define $\preceq$ to be exactly as 
$\leq$ except that we set $x_w$ to be the largest element.
Moreover, we extend the total order $\preceq$ to disjoint edges in a matching
$M$: given $\{x_1, x_2\}, \{y_1, y_2\} \in M$, we define
$\{x_1, x_2\} \preceq \{y_1, y_2\}$ if and only if $\min x_i \prec \min_i y_i$.

\paragraph{Lower levels.}
Now for $i = d-2 \dots 0$, we do the following: from level $i+1$, we
have an $(i+1)$-zero set $Z_{i+1}$. By~\Cref{lem:zero-set-to-perfect-matchings}, there exists a perfect
$i$-matching $M_i^*$ over $Z_{i+1}$. Notice that $M^*_i$ is an
$i$-zero matching.  We define $M_i$ as the shortest prefix (according
to $\preceq$) of $M^*_i$ that is $i$-dependent. Denote by $e_i$ the
 $\preceq$-maximum edge in $M_i$. Define the submatching
$L_i \subseteq M_i$ such that $L_i$ is an $i$-zero set. By construction,
$e_i \in L_i$. Finally, we set $Z_i=\bigcup L_i$ and notice that $Z_i$ is an
$i$-zero set.

\subsection{Characterizing an obstruction}
Our main proof strategy will be to show that an obstruction is unlikely to happen, implying that
our selected derived keys $\derive X$ are unlikely to be linearly dependent.
To get to that point, we first characterize such an obstruction in a way that
is independent of how it was derived in~\Cref{sec:construct}. With this classification
in hand, we will then be able to make a union bound over all possible obstructions.

\medskip\noindent
Corresponding to the two top levels, our obstruction consists of the following objects:
\begin{itemize}
	\item A set of keys $W \subseteq \Sigma^c$ of some size $w=|W|$.
	\item A special key $x_w \in W$, that we put last when we define $\preceq$.
	\item A maximal matching $M_d$  on $W \setminus \{x_{w}\}$.
	\item A perfect matching  $M_{d-1}$ on $W$, where $e_{d-1}$ is the only edge in $M_{d-1}$ incident to $x_w$.
	\item A submatching $L_{d-1}\subseteq M_{d-1}$, which includes $e_{d-1}$, and its support $Z_{d-1}=\bigcup L_{d-1}$.  
\end{itemize}

Note that the above objects do not include the whole selected set $X^{f,h}$,
the $d$-zero set $Z$, or the perfect matchings $M_{d-1}^*,M_d^*$ that we
used in our construction of the obstruction. Also, note that thus far,
the objects have not mentioned any relation to the hash function $h$.

To describe the lower levels, we need to define greedy matchings.  We
say a matching $M$ on a set $Y$ is \emph{greedy} with respect to the
total order $\preceq$ on $Y$ if either: (i) $M = \emptyset$ or; (ii)
the key $\min_\preceq Y$ is incident to some $e \in M$ and $M \setminus \{e\}$ is
greedy on $Y\setminus e$.
Assuming that $|Y|$ is even, we note that $M$ is greedy if and
only if it is a prefix of some $\preceq$-ordered perfect matching
$M^*$ on $Y$.

\medskip\noindent
For the lower levels $i=d-2,\ldots,1$, we then have the following
objects:
\begin{itemize}
	\item A greedy matching $M_i$ on $Z_{i+1}$. Denote with $e_{i}$ the $\preceq$-maximum edge in $M_i$.
	\item A submatching $L_i\subseteq M_i$ which includes $e_i$, and its support  $Z_i=\bigcup L_i$.
\end{itemize}

Note that the above objects describe all possible obstructions that we might construct. 
Moreover, any obstruction can be uniquely described as the tuple of objects
$$(W,x_w,M_{d},M_{d-1},e_{d-1},L_{d-1},Z_{d-1},\ldots	M_{1},e_{1},L_{1},Z_{1})\;.$$ 

\subsection{Confirming an obstruction}
In order for a given obstruction to actually occur, the tornado tabulation hash function $h=\toptab h\circ \derive h$ must satisfy the following conditions with respect to it:
\begin{itemize}
	\item The keys in $W$ are all selected, that is, $W \subseteq X^{f,h}$. 
	\item Either $W\subseteq Q$ or $Q\cup W\setminus\{x_w\}$ is $d$-independent.
	\item For $i=1,\ldots,d$, $M_i$ is an $i$-matching.
	\item $M_d$ is $(d-1)$-independent.
\end{itemize}
For $i=1,\ldots,d-1$, 
\begin{itemize}
	\item $M_i\setminus \{e_i\}$ is $(i-1)$-independent.
	\item The submatching $L_i$ is $i$-zero (implying that
	$Z_i=\bigcup L_i$ is an $i$-zero set).
\end{itemize}

When a hash function $h$ satisfies the above conditions, we say that it \emph{confirms} an obstruction.
We now need to argue two things. We need (1) to verify that our construction of
an obstruction from~\Cref{sec:construct} satisfies these conditions on $h$, and (2) that,  given a fixed obstruction, the probability
that a random $h$ confirms the obstruction is very small. This
is captured by the two lemmas below.
\begin{lemma}\label{lem:construct-correct}
	The obstruction constructed in~\Cref{sec:construct} satisfies the conditions on
	$h=\toptab h\circ \derive h$.
\end{lemma}
\begin{proof} Since $W\setminus\{x_w\}\subseteq Z\setminus \{x^*\}$, the choice of $Z$
implies that either $W\subseteq Q$ or
$Q\cup W\setminus\{x_w\}$ is $d$-independent.

The matchings $M_i$ were all submatchings
	of perfect $i$-matchings. We also picked $M_i$ is minimally $i$-dependent,
	but for an $i$-matching $M_i$ this also implies that $M_i$ is
	minimally $(i-1)$-dependent, and therefore $M_i\setminus \{e_i\}$ is $(i-1)$-independent. Finally, we constructed $L_i$ and $Z_i$ to be $i$-zero.
\end{proof}
\begin{lemma}\label{lem:event-probability}
	Given the objects of an obstruction, the probability that 
	$h=\toptab h\circ \derive h$  confirms the obstruction is at most
	\[\ld(\prod_{x\in W\setminus \{x_w\}} p_x \rd) \bigg / |\Sigma|^{w-2+\sum_{i=1}^{d-2}(|M_i|-1)}.\]
\end{lemma}
\begin{proof}
For simplicity, we group the conditions above in the following events: we define the event $\calC_S$ to be the event  that   $W \setminus\{x_w\} \subseteq X^{f,h}$ given that the set $W\setminus\{x_w\}$ is $d$-independent. Then, for each $i\in\set{1,\ldots, d-1}$, we define $\calC^{(i)}$ to be the event that $M_i\setminus e_i$ is an $i$-matching conditioned on the fact that it is $(i-1)$-independent. Finally, we define the last event $\calC^{(d)}$ to be the event that $M_{d}$ is a $d$-matching conditioned on the fact that it is $(d-1)$-independent. It is then sufficient to show that: 

$$\Pr\parentheses{\calC_S \wedge \bigwedge_{i=1}^{d} \calC^{(i)}} \leq 	\ld(\prod_{x\in W\setminus \{x_w\}} p_x \rd) \bigg / |\Sigma|^{w-2+\sum_{i=1}^{d-2}(|M_i|-1)}.$$

We proceed from the bottom up, in the following sense. For every $i\in\set{1,\ldots, d-1}$, the  randomness of the event $\calC^{(i)}$ conditioned on $\bigwedge_{j=1}^{i-1} \calC^{(j)}$  depends solely on $\derive h_i$. We invoke~\Cref{lem:independent-matchings} for $M_i\setminus e_i$ and get that 

$$\Pr\parentheses{\calC^{(i)} \conditioned \bigwedge_{j=1}^{i-1} \calC^{(j)}} \leq 1/\size{\Sigma}^{ \size{M_i}-1}\;.$$

When it comes to level $d$, from~\Cref{lem:independent-matchings}  applied to $M_d$, we get that the probability of $\calC^{(d)}$ conditioned on $\bigwedge_{j=1}^{d-1}\calC^{(j)}$ is at most $1/\size{\Sigma}^{\size{M_d}}$. We now notice that $\size{M_{d-1}}+ \size{M_d} = w-1$ by construction, and so:
\begin{align*}
	\Pr\parentheses{\bigwedge_{j=1}^{d}\calC^{(j)}} &\leq 1/\size{\Sigma}^{w-2+\sum_{i=1}^{d-2} (\size{M_i}-1)}\;.
\end{align*}

Finally, we know that either $W\subseteq Q$ or $Q\cup W\setminus\{x_w\}$ is $d$-independent.
If $W\subseteq Q$, then $p_x=1$ for all 
$x\in W$, and therefore, trivially,
$$\Pr\parentheses{\calC_S \conditioned \bigwedge_{j=1}^{d}\calC^{(j)} } \leq \prod_{x\in W\setminus \set{x_w}} p_x= 1\;.$$
Otherwise $Q\cup W\setminus\{x_w\}$ is $d$-independent. This means that
the derived keys $\derive h(Q\cup W\setminus\{x_w\})$ are linearly independent. We can then apply~\Cref{lem:simple-tab-on-lin-indep-sets} to this set of derived keys and the final simple tabulation hash function  $\toptab h$. We get that the final hash values $h(Q\cup W\setminus\{x_w\})$ are chosen independently and uniformly at random. For any one element $x$, by definition, $\Pr\parentheses{x \in X^{f,h}} \leq p_x$ and so:
$$\Pr\parentheses{\calC_S \conditioned \bigwedge_{j=1}^{d}\calC^{(j)} } \leq \prod_{x\in W\setminus \set{x_w}} p_x\;.$$
This completes the claim.
\end{proof}

\section{Simplified analysis}\label{sec:simplified}
In this section, we present a simplified analysis showing that,
under the hypotheses of our theorem, $\derive h(X^{f, h})$ is linearly
dependent with probability at most $\Theta(\mu^3 (17/|\Sigma|)^{d}) + 2^{-|\Sigma|/8}$.
Later, we will replace $(17/|\Sigma|)^d$ with
$(3/|\Sigma|)^d$, which essentially matches the growth in our lower bound.

For now we use a fixed limit $w_0=|\Sigma| / 2^{5/2}$ on the set size $w=|W|$. With this limited set size,
we will derive the $\Theta(\mu^3 (17/|\Sigma|)^d)$ bound. The $2^{-|\Sigma|/8}$ bound will
stem for sets of bigger size and will be derived in a quite different way.

Our goal is to study the probability that there exists a combinatorial
obstruction agreeing with a random tornado hash function $h$; if not,
$\derive h(X^{f, h})$ is linearly independent. To do this,
we consider a union bound
over all combinatorial obstructions as such:
\begin{equation}\label{eq:1}
	\sum_{W,x_w,M_{d},M_{d-1},e_{d-1},L_{d-1},Z_{d-1},\ldots
		M_{1},e_{1},L_{1},Z_{1}}\Pr_{h}[h\textnormal{ confirms the obstruction}]
\end{equation}
The above sum is informally written in that we assume that each element
respects the previous elements of the obstruction, e.g., for $i<d-2$, $M_i$ is a
greedy matching over
$Z_{i+1}$. Likewise, in the probability term, it is understood that $h$ is supposed to confirm the obstruction whose combinatorial description is $({W,x_w,M_{d},M_{d-1},e_{d-1},L_{d-1},Z_{d-1},\ldots
	M_{1},e_{1},L_{1},Z_{1}})$ .

\medskip\noindent
Using Lemma \ref{lem:event-probability}, we bound \req{eq:1} by
\begin{align}
	&\sum_{W,x_w,M_{d},M_{d-1},e_{d-1},L_{d-1},Z_{d-1},\ldots
		M_{1},e_{1},L_{1},Z_{1}}\ld(\prod_{x\in W\setminus \{x_w\}} p_x \rd) \bigg / |\Sigma|^{w-2+\sum_{i=1}^{d-2}(|M_i|-1)}\nonumber\\
	&\rule{1cm}{0ex}\leq \left(\sum_{W,x_w,M_{d},M_{d-1},e_{d-1},L_{d-1},Z_{d-1}}\ld(\prod_{x\in W\setminus \{x_w\}} p_x \rd) |\Sigma|^{2-w} 2^{w/4-1}\right)\label{top-factor}\\
	&\rule{2cm}{0ex}{{\times}}\prod_{i=1}^{d-2}\max_{Z_{i+1}}\ld(\sum_{M_{i},e_{i},L_{i},Z_{i}}|\Sigma|^{1-|M_{i}|}\bigg/2^{(|Z_{i+1}|-|Z_{i}|)/4}\rd)\label{low-factor}.
\end{align}
Above , $W,x_w,M_{d},M_{d-1},e_{d-1},L_{d-1},Z_{d-1}$ are all the
elements from the top two levels and we refer to \req{top-factor} as
the ``top'' factor. We refer to Equation
\req{low-factor} as the ``bottom'' factor which is the product of ``level'' factors.

For each lower level $i\leq d-2$, the elements
$M_{i},e_{i},L_{i},Z_{i}$ are only limited by $Z_{i+1}$, so for a
uniform bound, we just consider the maximizing choice of
$Z_{i+1}$. For this to be meaningful, we divided the level factor
\req{low-factor} by $2^{(|Z_{i+1}|-|Z_{i}|)/4}$ and multiplied
\req{top-factor} by $2^{w/4-1}\geq\prod_{i=1}^{d-2}2^{(|Z_{i+1}|-|Z_{i}|)/4}$.
This inequality follows because $|Z_{d-1}|\leq w$ and $|Z_1|\geq 4$.
The exponential decrease
in $|Z_{i+1}|$ helps ensuring that the maximizing choice of $Z_{i+1}$ has bounded
size (and is not infinite).
We note here that our bound $|W|\leq w_0=|\Sigma| / 2^{5/2}$ is only needed when
bounding the level factors,  where it implies that $|Z_{i+1}|\leq w_0$.

\subsection{The top two levels and special indexing for matchings and zero sets}
We now wish to bound the top factor \req{top-factor} from the two
top levels. Below, we will first consider $w$ fixed. Later we will
sum over all relevant values of $w$.

We want to specify the $w$ keys in $W$ using the fact that the keys from $W\setminus\{x_w\}$
hash independently.  Thus, we claim that we only have to specify the 
set $V=W\setminus\{x_w\}$, getting $x_w$ for free. 
By construction, we
have $x_w$ in the zero-set $Z_{d-1}$, so
\begin{equation}\label{eq:code-key}
	x_w=\bigsd (Z_{d-1}\setminus
	\{x_w\}).
\end{equation}
To benefit from this zero-set equality, we need the special
ordering $\preceq$ from the construction. It uses the standard
ordering $\leq$ of $V$ since all these keys are known, and then it just
places the yet unknown key $x_w$ last.  The special ordering yields
and indexing $x_1\prec x_2\prec\cdots\prec x_w$ (this is not the order
in which we traversed the cycles in the construction except that $x_w$
is last in $W$ in both cases). Formally, we can now specify all the
$M_i,L_i,Z_i$ over the index set $\{1,\ldots,w\}$. For $i<w$,
we directly identify $x_i$ as the $i$th element in $V$. This way
we identify all $x_i\in Z_{d-1}\setminus\{x_w\}$, and then we
can finally compute the special last key $x_w$; meaning that
we completely resolve the correspondance between indices and keys
in $W$. As a result, we can bound \req{top-factor} by
\[\sum_{V:|V|=w-1}\ld(\prod_{x\in V} p_x \rd)\times
\sum_{M_{d},M_{d-1},e_{d-1},L_{d-1},Z_{d-1}}|\Sigma|^{2-w}2^{w/4-1}.\]
For the first part, with $v=w-1$, we have
\begin{equation}\label{eq:set-number}
	\sum_{V:\, |V|=v} \ld(\prod_{x\in V} p_x\rd)
	=\frac1{v!}\sum_{(x_1,\ldots,x_v)}\prod_{i=1}^v p_{x_i}
	\leq\frac1{v!}\prod_{i=1}^v (\sum_{x}p_{x})
	=\frac{\mu^v}{v!}.
\end{equation}
For the second part, we note $M_{d-1}$ and $M_d$ can both be chosen in
$(w-1)!!$ ways and we know that $e_{d-1}$ is the edge in $M_{d-1}$
incident to $w$. Since the submatching $L_{d-1}$ of $M_{d-1}$ contains
the known $e_{d-1}$, it can be chosen in at most $2^{|M_{d-1}|-1}=2^{w/2-1}$
ways. Putting this together, we bound \req{top-factor} by
\begin{align}
	&\frac{\mu^{w-1}}{(w-1)!}((w-1)!!)^2\,2^{w/2-1}|\Sigma|^{-w+2}2^{w/4-1}
	=\frac{\mu^{w-1}}{|\Sigma|^{w-2}}\cdot\frac{(w-1)!!}{(w-2)!!}\cdot 2^{3w/4-2}\nonumber\\
	&\leq\frac{\mu^3}{|\Sigma|^2}2^{4-w}\cdot\frac{(w-1)!!}{(w-2)!!}\cdot 2^{3w/4-2}
	=\frac{\mu^3}{|\Sigma|^2}\cdot\frac{(w-1)!!}{(w-2)!!}\cdot 2^{2-w/4}\label{eq:no1}
\end{align}
Above we got to the second line using our assumption that $\mu\leq |\Sigma|/2$.
This covers our top factor \req{top-factor}, including specifying the
set $Z_{d-1}$ that we need for lower levels.

\paragraph{Union bound over all relevant sizes.}
We will now sum our bound \req{eq:no1} for a fixed set size $w$ over all relevant
set sizes $w\leq w_0$, that is, $w=4,6,\ldots,w_0$.
The factor that depends
on $w$ is
\[f(w)=\frac{(w-1)!!}{(w-2)!!}/2^{w/4}.\]
We note that $f(w+2)=f(w)\frac{w+1}{\sqrt 2 w}$ and $\frac{w+1}{\sqrt 2 w}<4/5$ for $w\geq 8$, so
\[\sum_{\textnormal{Even }w=4}^{w_0} f(w)< f(4)+f(6)+5f(8)<4.15.\]
Thus we bound the top factor \req{top-factor} over all sets $W$ of size up to $w_0$ by
\begin{equation}\label{top-to-w0}
	\sum_{\textnormal{Even }w=4}^{w_0} \left(\frac{\mu^3}{|\Sigma|^2}\cdot \frac{(w-1)!!}{(w-2)!!}\cdot 2^{2-w/4}\right)<16.6 (\mu^3/|\Sigma|^2).
\end{equation}

\subsection{Lower levels with greedy matchings}
We now focus on a lower level $i\leq d-2$ where we will bound the
level factor
\begin{equation}\label{level-factor}
	\max_{Z_{i+1}}\ld(\sum_{M_{i},e_{i},L_{i},Z_{i}}|\Sigma|^{1-|M_{i}|}\bigg/2^{(|Z_{i+1}|-|Z_{i}|)/4}\rd).
\end{equation}
We want a bound of $O(1/|\Sigma|)$, implying a bound of
$O(1/|\Sigma|)^{d-2}$ for all the lower levels in \req{low-factor}.

All that matters for our bounds is the cardinalities of the different sets, and
we set $m_i=|M_i|$ and $z_i=|Z_i|$. For now we assume that $z_{i+1}\leq w_0$ and
$m_i$ are given.

\begin{lemma}\label{lem:greedy-count}
	With a given linear ordering over a set $S$ of size
	$n$,
	the number of greedy matchings of size $k$ over $S$ is $(n-1)^{\ull k}$.
\end{lemma}
\begin{proof}  
	We specify the edges one at the time. For greedy matchings, when we pick
	the $j$th edge, the first end-point is the smallest free point in $S$ and then there
	are $n-2j+1$ choices for its match, so the number of possible greedy
	matchings of size $k$  is $(n-1)^{\ull k}$.
\end{proof}
Recall that $e_i$ denotes the last edge in our greedy matching $M_i$
and that $e_i\in L_i$. In our case, we are only going to specify
$M_i'=M_i\setminus \{e_i\}$ and $L_i'=L_i\setminus \{e_i\}$. The
point is that if we know $M_i'$ and $L_i'$, then we can compute
$e_i$. More precisely, we know that $e_i$ is the next greedy
edge to be added to $M'_i$, so we know its first end-point $x$.
We also know that $Z_i=\bigcup L_i$ is a zero-set, so the
other end-point can be computed as key
\begin{equation}\label{eq:code-edge}
	y=\bigsd(\{x\}\cup \bigcup L_i')
\end{equation}
Above we note that even though the keys in $x$ and $L'_i$ are only specified
as indices, we know how to translate between keys and indices, so we
can compute the key $y$ and then translate it back to an index.

By Lemma \ref{lem:greedy-count}, we have
$(z_{i+1}-1)^{\ull {m_i-1}}$ choices for $M'_i$, and then
there are $2^{m_i-1}$ possible submatchings $L'_i\subset M_i'$.
The number of combinations for $M_i,e_i,L_i$ is thus bounded
by
\[(z_{i+1}-1)^{\ull {m_i-1}}\cdot 2^{m_i-1}.\]
We want to multiply this by 
\[|\Sigma|^{1-m_{i}}/2^{(z_{i+1}-z_{i})/4}.\]
Here $z_i=2|L_i|\leq 2m_i$, so we get a bound of 
\begin{equation}\label{eq:gen-level-factor}
	(z_{i+1}-1)^{\ull {m_i-1}}\cdot (2^{3/2}/|\Sigma|)^{m_i-1}/2^{z_{i+1}/4-1/2}\leq  2^{1/2}(2^{3/2}z_{i+1}/|\Sigma|)^{m_i-1}/2^{z_{i+1}/4}.
\end{equation}
We now note that
\[(2^{3/2}z_{i+1}/|\Sigma|)\leq 1/2.\]
since $z_{i+1}\leq w\leq w_0=|\Sigma|/2^{5/2}$. Having
this factor bounded below 1 is critical because it implies that the term decreases as $m_i$ grows.

Still keeping $z_{i+1}$ fixed, we will now sum over all possible
values of $m_i$. Since $M_i$ contains a zero-matching $L_i$ of size at least 2,
that is 2 edges covering 4 keys, we have that $m_i\geq 2$. Moreover, the terms of the summation are
halving, hence they sum to at most twice the initial bound, so we get
\begin{equation}\label{eq:m-gen}
	\sum_{m_i\geq 2}2^{1/2}(2^{3/2}z_{i+1}/|\Sigma|)^{m_i-1}/2^{z_{i+1}/4}
	\leq 2^{3/2} (2^{3/2}z_{i+1}/|\Sigma|)/2^{z_{i+1}/4}.
\end{equation}
The maximizing real value of $z_{i+1}$ is $4/\ln 2=5.77$, but $z_{i+1}$ also
has to be an even number, that is, either 4 or 6, and 6 yeilds the maximum
bound leading to the overall bound of $16.98/|\Sigma|$ for \req{low-factor}.
Together with our bound \req{top-to-w0},  we get a total bound of
\begin{equation}16.6 (\mu^3/|\Sigma|^2)\cdot(17/|\Sigma|)^{d-2}.\label{eq:sloppy-w}
\end{equation}

\subsection{Large set sizes}
We now consider sets $W$ of sizes $w>w_0$. In this
case, we will only consider the two top levels of the obstructions.
Recall that we have a cycle traversal $x_1,\ldots,x_w$. If
$w>w_0$, we only use the prefix $x_1,\ldots,x_{w_0}$, where we
require that $w_0$ is even. The two matchings $M_d$ and $M_{d-1}$
are reduced accordingly. Then $M_{d-1}$ is a perfect
matching on $W_0=\{x_1,\ldots,x_{w_0}\}$ while $M_d$
is maximal excluding $x_{w_0}$. Each of these matchings
can be chosen in $(w_0-1)!!$ ways.

This time we know that we have full independence. We know that
$\derive W_0$ is a strict subset of the original minimimal zero set
$\derive Z$ of derived keys, so the keys in $W_0$ all hash
independently. Therefore the probability that they are all selected is
bounded by $\prod_{x\in W_0} p_x$.

Also, since we terminate the traversal early, we know that
$M_{d-1}$ is $(d-2)$-linearly independent, and $M_{d}$ must be
$d-2$-linearly independent, so the probability that these
two matchings are satisfied is
$1/|\Sigma|^{|M_{d-1}|+|M_{d}|} \leq 1/|\Sigma|^{w_0-1}$.

The total bound for all $w>w_0=|\Sigma|/2^{5/2}$ is now
\begin{align}
	\sum_{W_0:\,|W_0|=w_0, M_d,M_{d-1}} \ld(\prod_{x\in W_0} p_x\rd)/|\Sigma|^{w_0-1}
	&\leq \frac{\mu^{w_0}}{w_0!}((w_0-1)!!)^2/|\Sigma|^{w_0-1}
	\leq w_0|\Sigma|/2^{w_0}\label{large-set}\\
	&\leq 1/2^{|\Sigma|/8}.
\end{align}
The last step holds easily for $|\Sigma|\geq 256$. We note that
\req{large-set} holds for any choice of $w_0$, limiting the
probability of any obstruction with $|W|<w$.

\section{Tighter analysis}\label{sec:tighter}
We will now tighten the analysis to prove that 

\maintheorem*

\subsection{Bottom analysis revisited}\label{sec:tight}
We first tigthen the analysis of the bottom factor \req{low-factor}
so as to get a bound of $O((3/\Sigma)^{d-2})$ and, together with our top bound \req{top-to-w0}, obtain an overall bound of  $O(\mu^3(3/\Sigma)^d)$ matching our
lower bound within a constant factor. We are still using our
assumption that $w\leq w_0\leq |\Sigma|/2^{5/2}$ implying that
$z_{i+1}\leq |\Sigma|/2^{5/2}$.

First we look at a single level $i$. For a tighter analysis of the level factor
\req{level-factor}, we partition into cases
depending on $z_{i+1}$ and to some degree on $z_i=2m_i$. Recall that
$z_{i+1}$ is given from the level above.

If $z_{i+1}\leq 6$, then we must have $z_i=z_{i+1}$, since we cannot split into two
zero sets each of size at least 4. This implies that the factor $2^{|Z_{i+1}|-|Z_i|}$
is just one. Also, in this case, $M_i=L_i$ must be a perfect
matching on $Z_{i+1}=Z_i$, and such a perfect matching can be in $(z_{i+1}-1)!!$ ways.
Thus, for given $z_{i+1}\leq 6$, we bound the level factor
by
\begin{equation}\label{non-split-level}
	(z_{i+1}-1)!!/|\Sigma|^{z_{i+1}/2-1}\leq 3/|\Sigma|.
\end{equation}
with equality for $z_{i+1}=4$. As a result, if $z_{d-1}\leq 6$, then we have already
achieved a bound of $(3/|\Sigma|)^{d-2}$ for the bottom factor.
\paragraph{Split levels.}
We now consider $z_{i+1}\geq 8$. Now it is possible that $Z_{i+1}$ can split into two sets
so that $z_i<z_{i+1}$. For a given $m_i$ and $z_{i+1}$, we already had the
bound \req{eq:gen-level-factor}
\begin{equation*}
	2^{1/2}(2^{3/2}z_{i+1}/|\Sigma|)^{m_i-1}/2^{z_{i+1}/4}.
\end{equation*}
Since $(2^{3/2}z_{i+1}/|\Sigma|)\leq 1/2$, the worst case is achieved when
$m_i=2$, but here we can do a bit better. In \req{eq:gen-level-factor}
we had a factor $2^{m_i-1}$ to specify the subset $L_i$ of $M_i$
containing $e_i$, but since $L_i$ should have size at least $4$, we
have $L_i=M_i$. Thus, for $m_i=2$ and given $z_{i+1}$, we improve
\req{eq:gen-level-factor} to
\[2(z_{i+1}/|\Sigma|)/2^{z_{i+1}/4}.\]
For $z_{i+1}\geq 8$, this is maximized with $z_{i+1}=8$. Thus for $m_i=2$ and $z_{i+1}\geq 8$,
the level factor \req{level-factor} is bounded by 
\[(2\cdot 8/|\Sigma|)/2^2=4/|\Sigma|.\]
Now, for $m_i\geq 3$, we just use the bound \req{eq:gen-level-factor}. As in
\req{eq:m-gen}, we get
\[\sum_{m_i\geq 3}2^{1/2}(2^{3/2}z_{i+1}/|\Sigma|)^{m_i-1}/2^{z_{i+1}/4}
\leq 2\cdot 2^{1/2}(2^{3/2}z_{i+1}/|\Sigma|)^{2}/2^{z_{i+1}/4}\]
Over the reals, this is maximized with
$z_{i+1}=2\cdot 4/\ln 2\approx 11.52$, but we want the
maximizing even $z_{i+1}$ which is $12$, and then we
get a bound of $1.6/|\Sigma|$.

We now want to bound the whole bottom factor in the case where $z_{d-1}\geq 8$.
Let $j$ be the lowest level with $z_{j+1}\geq 8$. For all lower levels, if any,
the level factor is bounded by $(3/|\Sigma|)$. Also, for all higher levels $i>j$,
we have $2m_i=z_i\geq 8$, hence $m_i\geq 4$, so the level factor for higher levels
is bounded by our last $1.6/|\Sigma|$. The worst case is the level $j$, where
we could get any $m_s$ (note that if $s=1$, we have no guarantee that $m_i\leq 2$),
hence we have to add the bound $4/|\Sigma|$ for $m_s=2$ with the bound $1.6|\Sigma|$
for $m_s\geq 3$, for a total bound of $5.6/|\Sigma|$.

Thus, for a given $j$, the bottom factor is bounded by
\[(3/|\Sigma|)^{j-1}(5.6/|\Sigma)(1.6/|\Sigma|)^{d-2-j}=
(3/|\Sigma|)^{d-2}*(5.6/3)*(1.6/3)^{d-2-j}.\]
Summing, this over $j=1,\ldots,d-2\geq 1$, we get a bound of at most
\begin{equation}\label{bottom-worst}
	(3/|\Sigma|)^{d-2}(5.6/3)(1/(1-1.6/3)-1)<4(3/|\Sigma|)^{d-2} \;.
\end{equation}
This is our bound for the whole bottom factor when we maximized with
$z_{d-1}\geq 8$. Since it is larger than our bound of $3/|\Sigma|)^{d-2}$ when we
maximized with $z_{d-1}\leq 6$, we conclude that it is also our bound if
we maximize over any value of $z_{d-1}$. In combination with our top factor,
$16.6 (\mu^3/|\Sigma|^2)$ from \req{top-to-w0}, we
get a combined bound of
\begin{equation}\label{small-sets}
	16.6(\mu/|\Sigma|^2)4(3/|\Sigma|)^{d-2}\leq 7\mu^3(3/|\Sigma|)^d.
\end{equation}

\subsection{Increasing the maximization range}
We will now show how to deal with larger sets up to size $w_0^+=0.63|\Sigma|$. The level
factor with a fixed $z_{i+1}$ and $m_i$ has the following 
tight version from \req{eq:gen-level-factor}
\begin{align*}
	&(z_{i+1}-1)^{\ull {m_i-1}}\cdot (2^{3/2}/|\Sigma|)^{m_i-1}/2^{z_{i+1}/4-1/2}\\
	&=2^{1/4} f(m_i-1,z_{i+1}-1)\textnormal{ where }f(x,y)=y^{\ull {x}}\cdot (2^{3/2}/|\Sigma|)^{x}/2^{y/4}.
\end{align*}
We want to find the sum over relevant $m_i$ with the maximizing $z_i$. However,
we already considered all even $z_{i+1}\leq |\Sigma|/2^{5/2}$, so it
suffices to consider $z_{i+1}\in (|\Sigma|/2^{5/2},\,|\Sigma|/2]$.
Also, for a given $z_{i+1}$, we have to sum over all
$m_i\in [2,z_{i+1}/2]$. Thus we want to bound
\[\max_{\textnormal{odd }y\in [|\Sigma|/2^{5/2},\,|\Sigma|/2)}\sum_{x=1}^{\lfloor y/2\rfloor} f(x,y).\]
Note that $f(x+1,y)=f(x,y)\cdot (y-2x)(2^{3/2}/|\Sigma|)$. Hence
$f(x+1,y)<f(x,y)\iff y-2x<|\Sigma|/2^{3/2}$.

Assume that $y\geq |\Sigma|/2^{3/2}$ and consider the smallest integer $x_y$
such that $y-2x_y<|\Sigma|/2^{3/2}$. For a given value of $y$ this $x_y$
maximizes $f(x_y,y)$.

To bound $f(x_y,y)$, we note that 
\[y^{\ull{x}}<(y-x+1)^x\]
We have $y<w_0^+=0.63|\Sigma|$ and $y-2x_y<|\Sigma|/2^{3/2}$, so
\[y-x_y+1\leq (0.63+1/2^{3/2})|\Sigma|/2+1\leq |\Sigma|/2.\]
Therefore 
\[f(x_y,y)\leq y^{\ull{x_y}}\cdot (2^{3/2}/|\Sigma|)^{x_y}/2^{y/4}
\leq ((y-x_y+1)2^{3/2}/|\Sigma|)^{x_y}/2^{y/4}
\leq 1/2^{(y-2x_y)/4}.\]
We also have $y-2(x_y-1)\geq |\Sigma|/2^{3/2}$, so 
we get
\[f(x_y,y)\leq 1/2^{(|\Sigma|/2^{3/2}-2)/4}=1/2^{|\Sigma|/2^{7/2}-1/2}.\]
For $|\Sigma|\geq 256$, this is below $0.015/|\Sigma|^2$, so even if
we sum over all $x\leq y/2\leq|\Sigma|/2$, we get a bound below
$0.008/|\Sigma|$, that is,
\[\max_{y\in [|\Sigma|/2^{3/2},\,w_0^+]}\sum_{x=1}^{\lfloor
	y/2\rfloor} f(x,y)\leq 0.08/|\Sigma|.\]

Now consider $y<|\Sigma|/2^{3/2}$. Then $f(x,y)$ is decreasing in $x$
starting $f(0,y)=1/2^{y/4}$. Summing over all $x\leq y/2$,
we get $(y/2)/2^{y/4}$.
This expression is maximized with $y=4/\ln 2$, so for $y\geq |\Sigma|/4$,
we get a maximum of 
$(|\Sigma|/8)/2^{|\Sigma|/16}$. Now $2^{|\Sigma|/16}\geq |\Sigma|^2$
for $|\Sigma|\geq 256$, so we end up with
\[\max_{y\in [|\Sigma|/4| ,|\Sigma|/2^{3/2}]}\sum_{x=1}^{\lfloor y/2\rfloor}
f(x,y)\leq 1/(8|\Sigma|).\]
Finally we consider $y\in [|\Sigma|/2^{5/2},\,|\Sigma|/4]$.
Then
$f(x+1,y)\leq f(x,y) (y\, 2^{3/2} /|\Sigma|) \leq f(x,y)/2^{1/2}$, so
\[\sum_{x=0}^{\infty} f(x,y)\leq f(0,y)/(1-2^{-1/2})
\leq 1/((1-2^{-1/2})2^{y/4}).\]
With $y\geq |\Sigma|/2^{5/2}$, this is maximized with
$y=|\Sigma|/2^{5/2}$, so we  get the bound
\[\max_{y\in [|\Sigma|/2^{5/2},\,|\Sigma|/4]}\sum_{x=0}^{\infty} f(x,y)\leq 1/((1-2^{-1/2})2^{(|\Sigma|/2^{9/2})})\leq 0.344/|\Sigma|.\]
Putting all our bounds together, we conclude that 
\[\max_{\textnormal{odd }y\in [|\Sigma|/2^{5/2},\,w^+_0)}\sum_{x=1}^{\lfloor y/2\rfloor} f(x,y)\leq 0.344/|\Sigma|.\]
Our bound for the level factor with $z_{i+1}<(|\Sigma|/2^{5/2},|\Sigma|/2]$
is $2^{1/2}$ times bigger, but this is still much smaller than all the bounds
from Section \ref{sec:tight} based on smaller $z_{i+1}$. Thus we
conclude that the analysis from Section \ref{sec:tight} is valid
even if allow $z_{i+1}$ to go the whole way up to $w_0^+=0.63 |\Sigma|$.
Therefore \req{small-sets} bounds the probability of any obstruction
with set size $|W|\leq w_0^+$.

However, for the probability of obstructions with sets sizes $|W|>w_0^+$,
we can apply \req{large-set}, concluding that they happen with
probability bounded by 
\[w_0^+|\Sigma|/2^{w_0^+}=0.63|\Sigma|^2/2^{0.63|\Sigma|} <1/2^{|\Sigma|/2}.\]
The last step used $|\Sigma|\geq 256$. Together with \req{small-sets} we
have thus proved that the probability of any obstruction is bounded by
\begin{equation}\label{eq:simple}
	7\mu^3(3/|\Sigma|)^d+1/2^{|\Sigma|/2}.
\end{equation}
This is for simple tornado hashing without the twist.
\subsection{Tornado hashing including the twist}\label{sec:add-twist}
We will now return to the original tornado hashing with the twisted
character
\[\derive x[c]=x[c]\oplus \derive h_0(x[<c]).\]
This twist does not increase the number of lookups: it is still
$c+d$ with $c$ input characters and $d$ derived characters, so the
speed should be almost the same, but we will gain a factor $3/|\Sigma|$
in the probability, like getting an extra derived character for free.

The obstruction is constructed exactly as before except that we continue down to
level $0$ rather than stopping at level $1$. All definitions for level $i$
are now just applied for $i=0$ as well. However, we need to reconsider
some parts of the analysis. First, we need to prove that~\Cref{lem:independent-matchings} also holds for $i=0$:
\paragraph{Lemma \ref{lem:independent-matchings}} {\it Let $M$ be a partial matching on $\Sigma^c$. Conditioning on $M$
	being $(i-1)$-independent, $M$ is an $i$-matching with 
	probability $1/|\Sigma|^{|M|}$.}
\begin{proof}
	Since $M$ is
	$(-1)$-independent, we know that
	\[\DiffKeys(M,-1)=\left\{(\derive x\sd\; \derive y)[< c]\mid \{x,y\}\in 
	M\right\}\]
	is linearly independent. We note here that
	$(\derive x\;\sd\; \derive y)[< c]=(x\;\sd\; y)[< c]$. We want to
	know the probability that $M$ is a 0-matching, that is, the probability that $\tilde x[c]=\tilde y[c]$ for
	each $\{x,y\}\in M$. For $i>0$, we had
	\[\tilde x[c+i]=\tilde y[c+i]\iff \derive
	h_i((\derive x\sd\; \derive y)[<c+i])=0.\]
	However, for $i=0$, we have 
	\[\tilde x[c]=\tilde y[c]\iff 
	\tilde h_0((\derive x\sd\; \derive y)[<c])=x[c]\oplus y[c].\]
	However, Lemma \ref{lem:simple-tab-on-lin-indep-sets} states
	that the simple tabulation hash function $\tilde h_0$ is
	fully random on the linearly indepedent $\DiffKeys(M,-1)$, so
	we still conclude that $M$ is a 0-matching with 
	probability $1/\Sigma^{|M|}$.
\end{proof}
There is one more thing to consider. We have generally used that
if $Z$ is an $i$-zero set, that is, if $\derive Z[\leq c+i]$ is a zero-set, then $Z$ is also a zero-set. However, this may no
longer be the case. All we know is that $\derive Z[\leq c]$ is a zero-set.
This also implies that $Z[<c]=\derive Z[<c]$ is a zero-set. However,
we claim that
\begin{equation}\label{eq:oplus}
	\bigoplus Z=0
\end{equation}
Here  $\bigoplus Z$ is xoring that keys in $Z$ viewed as bit-strings.
Note that $\bigsd Z=\emptyset$ implies $\bigoplus Z=0$. Since
$Z[<c]$ is a zero set, we know that $\{\derive h_0(x[<c])\mid x\in Z\}$
is a zero set. We also know that $\derive Z[c]$ is
a zero set and it is equal to $\{\derive h_0(x[<c])\oplus x[c]\mid x\in Z\}$.
Thus we have 
\[\bigoplus \{\derive h_0(x[<c])\oplus x[c]\mid x\in Z\}=0=
\bigoplus \{\derive h_0(x[<c])\mid x\in Z\}\]
implying $\bigoplus \{x[c]\mid x\in Z\}=0$. Together with $Z[<c]$ being
a zero set, this settles \req{eq:oplus}. As a result, for the
coding key coding in \req{eq:code-key} and \req{eq:code-edge}, we just
have to replace $\bigsd$ with $\bigoplus$.

No other changes are needed. Level $0$ gives exactly the same level
factor \req{level-factor} as the levels $i>0$, so it is like getting
an extra level for free. Therefore with tornado hashing we
improve the bottom factor for simple tornado hashing \req{bottom-worst}
to $4(3/|\Sigma|)^{d-1}$ and the probability of any small obstruction
from \req{small-sets} to $7\mu^3(3/|\Sigma|)^{d+1}$. The probability of
any obstruction is thus improved from \req{eq:simple} to 
\begin{equation}
	7\mu^3(3/|\Sigma|)^{d+1}+1/2^{|\Sigma|/2}.
\end{equation}

\subsection{Full randomness for large sets of selected keys}
While implementing tabulation hashing we often want to set the character size so that all tables fit in cache. Indeed, this is the main reason why tabulation-based hashing schemes are extremely fast in practice.
However, restricting the size of $\Sigma$ constrains how many keys we can expect to be hashed fully randomly: in particular, \Cref{thm:intro-random-set} states that a set of characters $\Sigma$ allows for up to $|\Sigma| / 2$ selected keys to hash fully randomly with high probability. 
Most experiments on tabulation-based hashing \cite{Aamand0KKRT22, aamand2020fast}, though,  use 8-bit characters (namely $|\Sigma| = 2^8$). This implies that whenever the selected keys are $s \leq 2^7$ then they hash uniformly at random with high probability. However, we may want local full randomness for $s$ keys, where $s\gg 2^7$.
A trade-off between memory usage and number of keys hashed uniformly is of course unavoidable, however we can improve over the one in \Cref{thm:intro-random-set} with a clever observation.

More precisely, \Cref{thm:intro-random-set} states that given a set $X^{f, h}$ of query-selected keys with $\mu^f \leq |\Sigma| / 2$ their derived keys are linearly dependent with probability at most $\DepProb(\mu, d, \Sigma)=\DP$ while tornado is using $O(c|\Sigma|)$ words of memory and exactly $c + d$ lookups in tables of size $|\Sigma|$. Recall that $\mu^f$ here is the expected size of $X^{f, h}$ for a fully random $h$, when the queries are chosen by an adaptive adversary.
If we want to handle larger sets of selected keys with $\mu^f \gg |\Sigma|/2$ using \Cref{thm:intro-random-set}, then we need to use $O(c+d)$ larger tables of size roughly $2 \mu^f$. The clever observation we make here is that, in order to obtain a meaningful probability bound, it is not necessary at all to employ $O(c+d)$ of such larger tables. Indeed, we just need the tables in the top two levels to have size $|\Psi| \geq 2 \mu^f$ to obtain that the derived keys are linearly dependent with probability at most

	$$14 |X|^3(3/|\Psi|)^2(3/|\Sigma|)^{d-1} + 1/2^{\size{\Sigma}/2}\;$$
losing a factor two with respect to \Cref{thm:intro-random-set}. We name this variant of tornado tabulation tornado-mix tabulation, because the last two derived characters can be evaluated in parallel as in mixed tabulation hashing \cite{dahlgaard15k-partitions}.

\paragraph{Formal definition of tornado-mix.} 
For each $i = 0,\ldots, d-2$ we let $\derive h_i : \Sigma^{ c+i - 1} \longrightarrow \Sigma$ be a simple tabulation hash function. For $i = d-1, d$ we let $\derive h_i : \Sigma^{ c+d - 2} \longrightarrow \Psi$ be a simple tabulation hash function.
Given a key $x \in \Sigma^c$, we define its \emph{derived key} $\derive x \in \Sigma^{c+d-2}\times\Psi^2$ as $\derive x=\derive x_1\cdots \derive x_{c+d} $, where
\begin{equation*}
	\derive x_i = \begin{cases}
		x_i &\text{if $i = 1, \ldots,c-1$} \\
		x_{c} \xor \derive h_{0}(\derive x_1 \dots \derive x_{c-1}) &\text{if $i = c$} \\
	\derive h_{i-c}(\derive x_1 \dots \derive x_{i-1}) &\text{if $i=c+1,\ldots,c+d-2$}.\\
	\derive h_{i-c}(\derive x_1 \dots \derive x_{c+d-2}) &\text{if $i=c+d-1, c+d$}.
	\end{cases}
\end{equation*}
Finally, we have a simple tabulation hash function $\widehat h: \Sigma^{c+d-2}\times\Psi^2 \longrightarrow \cR$, that we apply to the derived key.
The \emph{tornado-mix tabulation} hash function $h: \Sigma^c \longrightarrow \cR$ is then defined as $h(x) = \widehat h(\derive x)$.

\paragraph{Cache efficiency.} It is worth to notice that such construction allows us to store in cache all $|\Sigma|$-sized tables while the two $|\Psi|$-sized tables might overflow cache. However, these larger tables are accessed only once while evaluating tornado and they can be accessed in parallel.

\paragraph{Local uniformity theorem.} Now we are ready to state an analogous of \Cref{thm:tech-random-set} for larger sets of selected keys. In what follows we use $f$, $\mu^f$ and $X^{f, h}$ as defined in \Cref{sec:technical-results}.

\begin{restatable}{theorem}{tornadomixtheorem}\label{thm:tornadomix}
	Let $h=\widehat h\circ \derive h:\Sigma^c\to\cR$
	be a random tornado-mix tabulation hash function with $d$ derived characters, the last two from $\Psi$, and select function $f$. If $\mu=\mu^f \leq \size\Psi/ 2$
	then the derived selected keys $\derive h(X^{f,h})$ are linearly
	dependent with probability at most
$$14 \mu^3(3/|\Psi|)^2(3/|\Sigma|)^{d-1} + 1/2^{\size{\Sigma}/2}\;.$$
\end{restatable}

\begin{proof}
This proof works exactly as the proof of \Cref{thm:tech-random-set}, except for a few slight differences. We limit ourselves to listing such small differences.
In \Cref{sec:construct} we define $Z$ as the smallest $d$-zero set among those $d$-zero sets minimizing the number of elements not in $Q$. This definition implies that there exists $x^* \in Z$ such that $Z \setminus \{x^*\}$ is $d$-independent. Moreover, either $Z \subseteq Q$ or we can choose $x^* \in Z \setminus Q$. In the original proof, we considered the alternating-cycle structure induced by the top level matchings $M_{d-1}^*$ and $M_d^*$ and traversed these cycles leaving $x^*$ last. This ensured that the edges from $M^*_d$ were always ``surprising'' in the sense that the probability of any such edge being realised by our random choice of $h$ was $1/|\Sigma|$, even after conditioning on all previously discovered edges (these events were, indeed, independent).
This observation allowed us to bound the probability of our obstruction being realised by $h$. In fact, we wanted our obstruction to be constituted by edges which realisations were independent, exluding the last edge. This is exactly what we did in \Cref{sec:construct}, where all edges but the last edge $e_{d-1} \in M^*_{d-1}$ were realised by $h$ independently.

In the current scenario, it is not obvious that the realisations of  traversed edges from $M^*_d$ are all independent and the first edge introducing a dependence belongs to $M_{d-1}^*$. Here, instead, we traverse the alternating-cycle graph until we find a prefix $W= \{x_1 \dots x_w\}$ such that either (i) $M^*_{d-1}$ restricted to $W$ or (ii) $M^*_d$ restricted to $W$ is $(d-2)$-dependent. Case (i) is identical to the case already analysed, but in case (ii) we need some small changes, essentially swapping the roles of $M_{d}$ and $M_{d-1}$. 

To understand the interplay between the two cases, note that every
time we add a vertex $x_w$, we add an edge to $M_{d-1}^*|_W$ if $w$ is odd, and to $M_{d}^*|_W$ if $w$ is even. Case (i) can only happen if
$w$ is even while Case (ii) can only happen if $w$ is odd. If we get to $x_w=x^*$, then we have have $W=Z$ and then we are in case (i) since $Z$ is even. Thus, if
we end in case (ii), then $x^*\not\in W$.

We now consider the slightly reduced traversal sequence where we simply drop the first vertex in the last cycle considered in the 
traversal. The point is that this vertex was not matched in 
$M_{d}^*|W$, but if we skip it
then $M_{d}^*|W$ is
a perfect $d-2$-dependent matching while $M_{d-1}^*|W$ is a maximal matching. Since we did not have $x^*$ in $W$, we have $W \subseteq Z\setminus\set{x^*}$ implying full independence of the hash values over $W$. Thus, using this reduced traversal sequence, we have swapped the roles of the two top levels. 
Since we now have two cases, our union based probability bounds are doubled (increasing the leading
constant from 7 to 14).

The rest of the analysis is unchanged except that two top levels use the different alphabet $\Psi$. This means 
$\Psi$ replaces $\Sigma$ in our bound \req{eq:no1} for the two top levels. This implies that our
overall bound is multiplied by
$\abs{\Sigma}^2/\abs{\Psi}^2$. Combined with the doubling, we get an overall probability bound of
	$$14 |X|^3(3/|\Psi|)^2(3/|\Sigma|)^{d-1} + 1/2^{\size{\Sigma}/2}\;.$$
\end{proof}

From \Cref{thm:tornadomix}, using \Cref{lem:simple-tab-on-lin-indep-sets-regular}, we derive the following analogous of \Cref{thm:intro-random-set} which was already stated in the introduction. Here we use the same notation as in \Cref{thm:intro-random-set}, where $h^{(s)}$ are the selection bits and $h^{(t)}$ are the free bits.

\tornadomixrandomtheorem*

\section{Upper Tail Chernoff}\label{sec:upper-tail-chernoff}

In this section, we show a Chernoff-style bound on the number of the selected keys $X^{f,h}$. 

\upperchernoff*
\begin{proof}
    We follow the proof of the Chernoff bound for the upper tail.
	Mainly, we let $\mathcal{J}_x$ denote the indicator random variable for whether the key $x$ gets selected in $X^{f,h}$. Then $\size{X^{f,h}}=\sum_{x\in \Sigma^c} \mathcal{J}_x$.
	For simplicity, we let $a=1+\delta$. Then for any $s>0$: 
    \begin{align*}
        \Prp{\size{X^{f,h}} \geq a\cdot \mu^f \wedge \mathcal{I}_{X^{f,h}} }
        &= \Prp{ \size{X^{f,h}} \cdot \indicator{\mathcal{I}_{X^{f, h}}} \geq a\cdot \mu^f }\\
        &= \Prp{ e^{s\size{X^{f,h}} \cdot \indicator{\mathcal{I}_{X^{f, h}}} } \geq e^{ s a \cdot \mu^f} }\\
        &\leq \frac{\Ep{M_{ \size{X^{f,h}}\cdot \indicator{\mathcal{I}_{X^{f, h}}} }(s) }}{e^{sa \cdot \mu^f}} \;,
    \end{align*}
	where $M_{ \size{X^{f,h}}\cdot \indicator{\mathcal{I}_{X^{f, h}}} }(s)$ is the moment generating function of the random variable $(\size{X^{f,h}}\cdot \indicator{\mathcal{I}_{X^{f, h}}})$, and is equal to:
	$$ M_{ \size{X^{f,h}}\cdot \indicator{\mathcal{I}_{X^{f, h}}} }(s) = \Ep{e^{s\cdot \size{X^{f,h}} \cdot \indicator{\mathcal{I}_{X^{f, h}}}} } = \sum_{i=0}^{\infty}{\frac{s^i}{i!} \cdot \Ep{\size{X^{f,h}}^i \cdot \indicator{\mathcal{I}_{X^{f, h}}} }}\;.$$
	
	\medskip\noindent
	Since $\size{X^{f,h}}=\sum_{x\in \Sigma^c} \mathcal{J}_x$, each moment $\Ep{\size{X^{f,h}}^i \cdot \indicator{\mathcal{I}_{X^{f, h}}} }$ can be written as the sum of expectations of the form $\Ep{\prod_{x\in S}\mathcal{J}_{x} \cdot \indicator{\mathcal{I}_{X^{f, h}}} }$, where $S\subseteq \Sigma^c$ is a fixed subset of size $\leq i$. 
    In general, the moment generating function can be written as a sum, with positive coefficients, of terms of the form $\Ep{\prod_{x\in S}\mathcal{J}_{x} \cdot \indicator{\mathcal{I}_{X^{f, h}}}}$ for some fixed subset $S$. Moreover, each such term amounts to
	$$ \Ep{\prod_{x\in S}\mathcal{J}_{x} \cdot \indicator{\mathcal{I}_{X^{f,h}}} } = \Prp{S \subseteq X^{f,h} \wedge \mathcal{I}_{X^{f,h}}}\;.$$

	\medskip\noindent
	We now condition on the hash values of the query keys $h|_Q$.  For any set $S\subseteq \Sigma^c$, we  let $\mathcal{I}_{S\cup Q}$ denote the event that the derived keys in $\derive h(S\cup Q)$ are linearly independent. Note that $\mathcal{I}_{X^{f,h}}$ being true implies that  $\mathcal{I}_{S\cup Q}$  is true. Moreover, since $S \cup Q$ is a fixed (deterministic) set, the event $\mathcal{I}_{S\cup Q}$ only depends on the randomness of $\derive h$. We get the following:
	\begin{align*}
		\Prp{S \subseteq X^{f,h}\wedge  \mathcal{I}_{X^{f,h}} \conditioned h|_Q}&\leq 
		\Prp{S \subseteq X^{f,h}\wedge  \mathcal{I}_{S\cup Q} \conditioned h|_Q}\\
		&= E\Big[\indicator{ \mathcal{I}_{S \cup Q} } \cdot \Pr\parentheses{S \subseteq X^{f,h} \conditioned \derive h, h|_Q} \conditioned h|_Q\Big] \\
		&= E\Big[\prod_{x\in S}\mathcal{J^*}_{x} \conditioned h|_Q \Big] \;,
	\end{align*}
	
	where $\set{\mathcal{J^*}_{x}}_{x\in X}$ denotes the indicator random variables for choosing to select the keys in $S$ independently and uniformly at random when we fix the hash values of the query keys. 	This last step is due to the fact that the event if $\mathcal{I}_{S \cup Q}$ is true then, then  $(h|_S, h|_Q)$ is fully random and it has the same distribution as $(h^*|_S, h|_Q)$, where $h^*$ is a fully-random hash function.  If $\mathcal{I}_{S \cup Q}$ is false, then the entire expression is $0$. We can thus continue the proof of Chernoff's as if $\size{X^{f,h}}$ were a sum of independent random variables.  The claim follows by noticing that, when  $\mathcal{I}_{X^{f,h}}$ is true, we have that $E\ld[ X^{f,h}\rd] \leq \mu^f$.
	
\end{proof}

\subsection{Upper Tail Chernoff for larger $\mu^f$}

We now consider the case in which we have a selector function for which $\mu^f > \size{\Sigma}/2$. In this case, even though we cannot guarantee that the set of derived selected keys is linearly independent whp, we show that, whp, its size still cannot be much larger than $\mu^f$. This particular case will be useful in the analysis of linear probing from~\Cref{sec:linear-probing}.

\begin{restatable}{lemma}{largeupperchernoff}\label{lemma:large-upper-chernoff}
	Let $h=\widehat h\circ \derive h:\Sigma^c\to\cR$
	be a random tornado tabulation hash function with $d$ derived characters, query key $Q$ with $\size{Q}<\size{\Sigma}/2$, and selector
	function $f$ such that $\mu^f>\size{\Sigma}/2$.  
	Then, for any $\delta>0$,  the set of selected derived keys $X^{f,h}$ satisfies the following:
	$$ \Pr\ld[{\size{X^{f,h}} \geq (1+\delta)\cdot  \mu^f } \rd]\leq 4\cdot \parentheses{\frac{e^{\delta_0}}{(1+\delta_0)^{1+\delta_0}}}^{\size{\Sigma}/2} + 4\cdot \DepProb(\size{\Sigma}/2,d,\Sigma) \;,$$
	where $$\delta_0 =\frac{\mu^f}{\mu^f-\size{Q}}\cdot \frac{\size{\Sigma}/2 - \size{Q}}{\size{\Sigma}/2} \cdot \delta \geq \parentheses{1-\frac{\size{Q}}{\size{\Sigma}/2} }\cdot \delta \;.$$
\end{restatable}

\begin{proof}
	We modify the given selector function $f$ to get another selector function $f_p$ with the same set of query keys $Q$ but with a much smaller $\mu^{f_p}$. Selection according to $f_p$ is done such that, once $f$ selects a key in $\Sigma^c\setminus Q$,  $f_p$  further sub-selects it with some probability $p$. This sub-selection is done independently for every selected key. 
	It follows that, for all $x\in \Sigma^c\setminus Q$,  $p_x^{f_p}=p_x^f \cdot p$. Taking into account query keys, we also get that $\mu^{f_p}=(\mu^f -\size{Q})\cdot p + \size{Q}=p\mu^f+(1-p)\size{Q}$.
	
	Moreover, one can show that, as long as $\mu^{f_p} \leq \size{\Sigma}/2$,  all our results about linear independence also
	hold for such sub-sampled select function. In particular, the only aspect of the proof of~\Cref{thm:tech-random-set} that depends on the probabilities $p_x^{f_p}$ is the proof of~\Cref{lem:event-probability}. There, we invoke~\Cref{lem:simple-tab-on-lin-indep-sets} to get an upper  bound on the probability that the set $W\setminus\{x_w\}$ is selected, given that it is $d$-independent. Notice that, if, additionally, we sub-sample elements from $W\setminus\{x_w\}$ each independently with probability $p$, we obtain the same bounds as if we initially selected elements with probability $p_x^{f_p}$. Therefore, when $\mu^f > \size{\Sigma}/2$, we can pick any $p\leq (\size{\Sigma}/2-\size{Q})/ ( \mu^f-\size{Q})$ to get that the set of derived keys for the sampled selection $\derive X^{f_p,h}$ is indeed linearly independent with probability at least $1-\DepProb( \mu^{f_p}, d, \Sigma)$. We use $I_{X^{f_p,h}}$ to denote this event.

	The next step is to notice that, conditioned on $\size{X^{f,h}}-\size{Q}$, the distribution of $\size{X^{f_p,h}}-\size{Q}$ is exactly the binomial distribution $B(\size{X^{f,h}}-\size{Q}, p)$. Then, for  $p> 1/(\size{X^{f,h}}-\size{Q})$, we have from~\cite{greenberg2014tight} that
	
	$$\Prpcond{\size{X^{f_p,h}} \geq \Epcond{\size{X^{f_p,h}}}{\size{X^{f,h}}}}{\size{X^{f,h}} } >1/4\;.$$
	\medskip\noindent
	Therefore, for any $t>0$: 
	$$\Prp{\size{X^{f_p,h}} \geq p\cdot t + (1-p)\size{Q} \conditioned \size{X^{f,h}} \geq t} > 1/4 \;.$$
	
	\medskip\noindent
	We now use this to derive an upper bound on $\Pr\parentheses{\size{X^{f,h}} \geq t}$ as such:
	
	\begin{align*}
		\Pr\parentheses{\size{X^{f,h} }\geq t  } & < 4 \cdot \Pr\parentheses{\size{X^{f_p,h}} \geq  p\cdot t + (1-p)\size{Q} \conditioned \size{X^{f,h}} \geq t} \cdot \Pr\parentheses{\size{X^{f,h} }\geq t  } \\
		&\leq  4 \cdot\Pr\parentheses{\size{X^{f_p,h}} \geq  p\cdot t + (1-p)\size{Q} } \\
		&\leq 4 \cdot\Pr\parentheses{\size{X^{f_p,h}} \geq  p\cdot t + (1-p)\size{Q} \wedge  I_ {X^{f_p,h}} }  +\\
		&+ 4 \cdot\DepProb(\mu^{f_p},d, \size{\Sigma}) \;.
	\end{align*}

	\medskip\noindent
	We now plug in	$ t=(1+\delta)\cdot \mu^f$, and get that
	$	 p\cdot t + (1-p)\size{Q} \geq  (1+\delta_0)\cdot \mu^{f_p}$ for $$\delta_0 \leq \frac{p \mu^f}{\mu^{f_p}} \cdot \delta \;.$$
	\noindent
	We then invoke~\Cref{lemma:upper-chernoff} to get that
	
	\begin{align*}
		\Pr\parentheses{\size{X^{f_p,h}} \geq (1+\delta_0)\cdot  \mu^{f_p}  \wedge  I_ {X^{f_p,h}} } &\leq \parentheses{\frac{e^\delta_0}{(1+\delta_0)^{1+\delta_0}}}^{ \mu^{f_p}}\;.
	\end{align*}
	
	\noindent
	Finally, we instantiate $p= (\size{\Sigma}/2-\size{Q})/ (\mu^f-\size{Q})$ such that $\mu^{f_p} = \size{\Sigma}/2$.  and notice that indeed, $p> 1/(\size{X^{f,h}}-\size{Q})$ when $\size{X^{f,h}} \geq (1+\delta)\cdot \mu^f$ and $\size{Q} < \size{\Sigma}/2$.  We notice that, in this case, $$\delta_0 =\frac{\mu^f}{\mu^f-\size{Q}}\cdot \frac{\size{\Sigma}/2 - \size{Q}}{\size{\Sigma}/2} \cdot \delta \geq \parentheses{1-\frac{\size{Q}}{\size{\Sigma}/2} }\cdot \delta \;.$$
	The argument follows.
	
\end{proof}
\section{Linear Probing with Tornado Tabulation}\label{sec:linear-probing}

In this section we show how to formally apply our framework to obtain results on linear probing with tornado tabulation.
We present the following main result comparing the performance of linear probing with tornado tabulation to that of linear probing using fully random hashing on a slightly larger keyset.

\begin{theorem}\label{thm:linearprobing}
  Let \(S, S^* \subseteq \Sigma^c\) be sets of keys of size \(n\) and \(n^* = (1+15\sqrt{\log(1/\delta)/\ssigma})n\), respectively, for some \(\delta \in (0,1/6)\).
  Let $T$, $T^*$ be arrays of size \(m\), a power of two. Now consider inserting the keys in \(S\) (\(S^*\)) into $T$ ($T^*$)  with linear probing using tornado tabulation (fully random hashing).
Let $X$ and $X^*$ be the number of comparisons performed when inserting a new key \(x\) in each of \(T\) and \(T^*\) (i.e., $x\notin S\cup S^*$).
Given the restrictions listed below there exists an event $\calE$ with \(\Pr(\calE) \geq 1 - \ld(1/\ssigma + 6\delta + 61 \log n \cdot  \DepProb(\ssigma/2, d, \Sigma)\rd)\) such that, conditioned on $\calE$, \(X\) is stochastically dominated by \(X^*\). \\
Restrictions:
\begin{itemize}
\item \(n/m \leq 4/5\)
\item \(\ssigma \geq 2^{16}\)
\item \(\ssigma \geq 30 \cdot \log n\)
\item \(\sqrt{\log (1/\delta)/\ssigma} \leq 1/18\) 
\end{itemize}
\end{theorem}

From~\Cref{thm:linearprobing}, it follows that linear probing using tornado tabulation achieves the same expected number of comparison as in the fully random setting,
a proof is given in \cref{subsec:linearprobing-expectation}.

\begin{corollary}\label{cor:linearprobing}
  Setting \(\delta = \Theta(1/\ssigma)\), \(d \geq 5\), \(\log n \leq o(\ssigma)\) in \cref{thm:linearprobing} we have
  \[
  \Ep{X} \leq \Ep{X^*} + o(1)\, .
  \]
\end{corollary}

The result of \cref{cor:linearprobing} is to be contrasted with previous work on practical implementations of linear probing.
While Knuth's analysis serves as evidence of linear probing's efficieny in terms of the number of comparisons performed, 
the advantage of linear probing (over other hash table implementations) is that each sequential memory access is much faster than
the random memory access  we do  for the first probe at $T[h(x)]$. How much faster depends on the computer system as does the cost of increasing the memory to reduce the load.  
Some experimental 
studies~\cite{black98linprobe,heileman05linprobe,thorup11timerev} have
found linear probing to be the fastest hash table organization for
moderate load factors (30-70\%). If the load is higher, we could
double the table size.

However, using experimental benchmarks to decide the hash table organization is only meaningful if the experiments are representative of the key sets on which linear probing would be employed. Since fully random hashing cannot be efficiently implemented, we might resort to weaker hash functions for which there could be inputs leading to much worse performance than in the benchmarks. The sensitivity to a bad choice of the hash function
led \cite{heileman05linprobe} to advice \emph{against} linear
probing for general-purpose use.
Indeed, Mitzenmacher and Vadhan \cite{mitzenmacher08hash} have proved that 2-independent hashing
performs as well as fully random hashing if the input has enough entropy. However, \cite{thorup12kwise,patrascu10kwise-lb} have shown that with the standard  2-independent linear hashing scheme, if the input is a dense set (or more generally, a dense subset of an arithmetic sequence), then linear probing
becomes extremely unreliable and the expected probe length increases from Knuth's $\frac{1+1/\eps^2}2$ to $\Omega(\log n)$, while the best known upper bound in this case is $n^{o(1)}$~\cite{knudsen2016linear}.\footnote{We note that if we only know that the hash function is $2$-independent, then the lower bound for the expected probe length is $\Omega(\sqrt{n})$ and this is tight.~\cite{thorup12kwise,patrascu10kwise-lb}}

In a breakthrough result, Pagh, Pagh and Ru{\v z}i{\'c}~\cite{pagh07linprobe} showed that if we use 5-independent hashing and
the load gap $\eps=\Omega(1)$, then the expected probe length is 
constant. P{\v a}tra{\c s}cu  and Thorup~\cite{patrascu12charhash} generalized this to an $O(1/\eps^2)$ bound for arbitrary
$\eps$, and showed that this
also holds for simple tabulation
hashing. However, in both
cases, the analysis hides
unspecified large constants
in the $O$-notation. Thus, with
these hashing schemes,
there could still be inputs for
which linear probing performs, say,  10 times worse in expectation, and then we would be better off using chaining.

Our result is of a very different nature.
We show that whp, for any given query key $x$, the probe length 
corresponding to $h(x)$ when we use tornado tabulation hashing is stochastically dominated by the probe length in a linear probing table that hashes slightly more keys but uses fully random hashing. In particular, this implies that whp,  the expected
probe length with tornado tabulation hashing is only a factor $1+o(1)$ away from Knuth's $\frac{1+1/\eps^2}2$. We get
this result without having to revisit Knuth's analysis from \cite{knuth63linprobe}, but simply
because we know that we are almost as good as fully random hashing, in a local sense that is sufficient for bounding the probe length (see \cref{subsec:linearprobingproof}).

As a further consequence of our results, we get that any benchmarking with  random keys that we do in order to set system parameters will  be representative for all possible sets of input keys. Moreover, the fact that tornado tabulation hashing only needs locality to perform almost as well as fully random hashing means that our arguments also work for other variants of linear probing.  For instance, ones where the maintenance of the hash table prioritizes the keys depending on when they were inserted, as first suggested in 
\cite{amble1974ordered}. Examples of this include Robin hood hashing where keys
with the lowest hash value come first \cite{celis1985robin} or time-reversed linear probing \cite{thorup11timerev} where the latest arrival comes first. In all these cases, tornado tabulation hashing performs almost as well as with fully-random hashing.

\subsection{Proof of~\cref{thm:linearprobing}}\label{subsec:linearprobingproof}
We let $\alpha = n/m$ denote the fill of the hash table and $\eps = 1-\alpha$.
The basic combinatorial measure that we study and employ is the \emph{run length}: If cells \(T[a]\), \(T[a+1], \dots, T[b-1]\) are all occupied by elements from \(S\) but both \(T[a-1]\) and \(T[b]\) are freem these \(b-a\) cells are called a run of length \(b-a\).
Let \(R(x, S)\) be the length of the run intersecting \(T[h(x)]\) and
note that $R(x, S) +1$ is an upper bound on the number of comparisons needed to insert some element $y$ into the table when $y$ hashes to the same location as $x$.

Let \(\Delta\) be the largest power of two such that \(3\alpha\Delta + 1 \leq \ssigma/2\).
The following lemma, proven in \cref{subsec:proof-runlength}, gives an upper bound on the probability that \(h(x)\) intersects a long run.

\begin{lemma}\label{lemma:newrunlength}
  \[ \Prp{R(x, S) \geq \Delta} \leq \frac{1}{\ssigma} + 60 \log n \cdot \DepProb(\ssigma/2, d, \Sigma) \, . \]
\end{lemma}

Let \(\calA\) be the event \(\ld(R(x, S) \geq \Delta \rd)\).
Assuming \(\calA\), there exists at least one unoccupied cell in table \(T\) between \(T[h(x) - \Delta]\) and \(T[h(x)]\) and likewise between \(T[h(x)]\) and \(T[h(x) + \Delta]\).
Hence the insertion of \(x\) only depends on the distribution of the much smaller key-set \(\set{s \in S \, \big| \, \size{h(s) - h(x)} \leq \Delta}\).

The second step of our proof bounds the probability that tornado tabulation behaves like a fully random hash function when restricted to this small set of keys.
As \cref{thm:intro-random-set} doesn't apply for arbitrary intervals we will instead cover the necessary interval with three dyadic intervals.
Recall that a dyadic interval is an interval of the form $[j2^i, (j+1)2^i)$, where $i,j$ are integers.
  In the following we will exclusively consider a number of dyadic intervals all of length \(\Delta\).
  Let $I_C$ denote the dyadic interval that contains $h(x)$, and similarly let $I_R$ and $I_L$ denote the dyadic intervals to the left and right, respectively, of $I_C$.
  We further let $X_C$ be the set of keys in $S$ that hash into the interval $I_C$, i.e., $X_C = \set{x\in S \mid h(x)\in I_C}$, and similarly, $X_R$ and $X_L$ are the pre-image of $h$ in $I_R$ and $I_L$, respectively.
  Given \(\calA\), the distribution of \(X\) is completely determined by the distribution of the keys in \(X_L \cup X_C \cup X_R\) and \(h(x)\).

  The expected size of each preimage is \(\alpha \Delta\) and our choice of \(\Delta\) thus allows us to apply \cref{thm:intro-random-set} to all three intervals at once. Let \(\calB\) be the event that the keys hashing into these intervals are distributed independently:
  
  \begin{corollary}\label{lemma:dyadic-uniform}
    With probability at least \(1 - \DepProb(\ssigma/2, d, \Sigma)\), \(\tilde h(X_R \cup X_C \cup X_L \cup \set{x})\) is linearly independent, such that $h$ hashes the keys in $X_R \cup X_C \cup X_L \cup \set{x}$ independently and uniformly in their respective intervals.
  \end{corollary}

We now define the analogous terms in the fully random setting. We let $I^{*}_C$ denote the dyadic interval in $T^{*}$ that contains $\hstar(x)$ and $I^{*}_R$ and $I^{*}_L$ the right and left neighboring dyadic intervals. Similarly, we let $X^{*}_C$, $X^{*}_R$ and $X^{*}_L$ denote their preimages under $\hstar$. The following lemma compares the two experiments in terms of the sizes of these preimages, and is proven in \cref{subsec:comparison-fully}.

\begin{lemma}\label{lemma:comparison-fully}
  Let \(\calC\) be the event \(\size{X_L} \leq \size{X^*_L} \land \size{X_C} \leq \size{X^*_C} \land \size{X_R} \leq \size{X^*_R}\), then
  \[ \Prp{\calB \land \bar\calC} \leq 6\delta \]
\end{lemma}

Let \(\calC\) be the event that each of the preimages \(X_i\) contain at most as many elements as the corresponding preimage \(X_i^*\).
Finally, define \(\calE = \calA \cap \calB \cap \calC\).
We will now present a coupling \(\tilde X\) of \(X\) which, when conditioned on \(\calE\), satisfies \(\tilde X \leq X^*\).

For every realization of $\size{X_L}$, $ \size{X_C}$ and $\size{X_R}$, we consider the following random process: starting from an emtpy table of size $m$ and using the fully random $\hstar$, insert the first $\size{X_L}$ elements from $X^{*}_L$, then the first $\size{X_R}$ elements from $X^{*}_R$ and finally, the first $\size{X_C}$ elements from $X^{*}_C$ (we do not insert any more elements after this). Note that, conditioned on $\calC$, we have that it is possible to choose such elements (i.e., $\size{X_R} \leq \size{X^{*}_R}$ etc.). Now let $\tilde{X}$ denote the number of comparisons performed when inserting $x$ into the table at this point in time.

We now have that $X$ (defined for the tornado tabulation) is identically distributed as $\tilde{X}$ (defined for a fully random hash function).
This is because event $\calA$ implies that the distribution of $X$ only depends on $X_R$, $X_C$ and $X_L$.
Event $\calB$ further implies that, on these intervals, $h$ behaves like a fully random hash function.
Now note that $\tilde{X} \leq X^{*}$, since we can continue the random process and add the remaining keys in $S^{*}$ and this can only increase the number of comparisons required to insert $x$ (i.e., ``more is worse'').

Left is to compute the total probability that any of our required events fail:
\begin{align*}
  \Prp{\bar \calE} &= \Prp{\bar \calA \lor \bar \calB \lor \bar \calC} \\
  &= \Prp{(\bar \calA \lor \bar \calC) \land \calB} + \Prp{\bar \calB} \\
  &\leq \Prp{\bar \calA} + \Prp{\calB \land \bar \calC} + \Prp{\bar \calB} \\
  &\leq 1/\ssigma + 6\delta + 61 \cdot \log n \cdot \DepProb(\ssigma/2, d, \Sigma) \, .
\end{align*}
This concludes the proof.

\subsection{Proof of~\cref{lemma:comparison-fully}}\label{subsec:comparison-fully}
As \(\Delta\) is chosen to be the largest power of two such that \(3 \alpha \Delta \leq \ssigma/2\) we get \(\frac{\ssigma}{12 \alpha} \leq \Delta\).
Let \(t\) be a constant, to be decided later.
For each \(i \in \set{L, C, R}\) let \(\calE_i\) be the event (\(\size{X_i} \leq t\)) and \(\calE_i^*\) be the event (\(t \leq \size{X_i^*}\)).
\begin{align*}
  \Prp{\calB \land \bar \calC} &=
  \Prp{\calB \land \exists i \in \set{L, C, R} : \size{X_i} > \size{X_i^*}} \\
  &\leq \Prp{\calB \land \exists i \in \set{L, C, R} : \bar \calE_i \lor \bar \calE_i^*} \\
  &\leq \sum_{i \in \set{L, C, R}} \ld( \Prp{\calB \land \size{X_i} > t} + \Prp{\calB \land \size{X_i^*} < t} \rd) \\
  &\leq \sum_{i \in \set{L, C, R}} \ld( \Prp{\calB \land \size{X_i} > t} + \Prp{\size{X_i^*} < t} \rd) \\
\end{align*}

Let \(\mu = \Ep{\size{X_i}} = \Delta\alpha\),
\(k = \sqrt{3\log (1/\delta)/\mu}\) and \(t = (1+k) \mu\).
Applying the tail-bound of \cref{lemma:upper-chernoff} we have
\begin{align*}
  \Prp{\calB \land \size{X_i} > t} &= \Prp{\calB \land \size{X_i} > (1+k)\mu} \\
  &\leq \exp\ld(- k^2 \mu / 3\rd) \\
  &= \delta \, .
\end{align*}

As \(\mu = \alpha \Delta \geq \ssigma/12\) we have
\begin{align*}
  k &= \sqrt{3\log (1/\delta)/(\alpha \Delta)} \\
  &\leq \sqrt{36 \log (1/\delta) / \ssigma} \\
  &\leq 1/3 \, .
\end{align*}
As \(n^* = (1+15\sqrt{\log (1/\delta)/\ssigma})n \geq (1 + 2.5 k)n\), \(\mu^* = \Ep{\size{X_i^*}} \geq (1 + 2.5k) \mu\).
Next, let \(k^* = k \cdot \sqrt{2/3}\). Then
\begin{align*}
  \mu^* \cdot (1 - k^*) &\geq (1 + 2.5k) \cdot \mu \cdot (1 - k^*) \\
  &= (1 + 2.5k) \cdot \mu \cdot (1 - \sqrt{2/3} \cdot k) \\
  &\geq \mu \cdot (1 + k) \\
  &= t \, .
\end{align*}
Hence
\begin{align*}
  \Prp{\size{X_i^*} < t} &\leq \Prp{\size{X_i^*} < (1-k^*) \mu^*} \\
  &\leq \exp(- (k^*)^2 \mu^* / 2) \\
  &= \exp(- k^2 \mu^* / 3) \\
  &\leq \exp(- k^2 \mu / 3) \\
  &= \delta \, .
\end{align*}

Summing over the six cases we see \(  \Prp{\calB \land \exists i \in \set{L, C, R} : \size{X_i} > \size{X_i^*}} \leq 6\delta\).

\subsection{Proof of \cref{lemma:newrunlength}}\label{subsec:proof-runlength}

Our proof relies on the simple observation that if \(T[a]\) through \(T[b]\) are all occupied and the run \emph{starts} in \(T[a]\) (i.e. \(T[a-1]\) is free which excludes the possibility of prior positions spilling over), then the preimage \(h^{-1}([a,b]) = \set{s \in S \, | \, h(x) \in [a, b]}\) must have size at least \(\size{b-a}\).
It must also be the case that either (1) the preimage \(h^{-1}([a+1, b])\) has size at least \((b-a)\cdot (1-\gamma) \) or (2) the preimage \(h^{-1}([a,a])\) is of size at least \((b-a)\cdot \gamma\), for any parameter \(\gamma \geq 0\).

We can generalize this to consider a run starting in any position \(T[b]\) within some interval \(b \in [a, c]\) which continues through \(T[d]\), then either (1) \(\size{h^{-1}([c, d])} \geq (d-c)\cdot(1-\gamma)\) or (2) \(\size{h^{-1}([a,c])} \geq (d-c) \cdot \gamma\).
We will refer to \([a,c]\) as the start-interval and to \([c, d]\) as the long interval.

Our strategy is to make both of these events unlikely by balancing the size of the start-interval with the number of keys needed to fill up the long interval.
Larger difference \((d-c)\) allows for a larger start-interval.
With a collection of roughly \(\log_{1+\eps/(6\alpha)} m\) such start-intervals we cover all possible starting poisitions before \(T[h(x) - \Delta]\), ruling out the possibility that a run starting before \(T[h(x)-\Delta]\) reaches \(T[h(x)]\).
The same strategy applied once more rules out the possibility that the run intersecting \(T[h(x)]\) will continue through \(T[h(x)+\Delta]\).

For our proof we set \(\gamma = \eps/3\) where \(\eps = 1 - n/m\) is the fill gap of \(T\).
The first pair of intervals we consider is the long interval \(I_0 = [h(x)-\Delta, h(x)]\) and the start-interval \(I_0'\) of size \(\Delta \eps/(6\alpha)\) preceding \(I_0\).
Next follows the long interval \(I_1 = I_0' \cup I_0\) with start-interval \(I_1'\) of length \(\size{I_1} \eps/(6\alpha)\) preceding it, and so forth.
Let \(\Delta_i = \size{I_i}\) and \(\Delta_i' = \size{I_i'}\). Observe how \(\Delta_i = \Delta \cdot (1 + \eps/(6\alpha))^i\), bounding the number of needed interval-pairs at \(\log_{1+\eps/(6\alpha)} (m/\Delta)\).

Depending on the length of the interval being inspected we can apply either \cref{lemma:upper-chernoff} or \cref{lemma:large-upper-chernoff} to bound the probability that the preimage exceeds the given threshold.
Taking the maximum of the two bounds simplifies the analysis.
Letting \(X\) be the size of a preimage, \(\mu = \Ep{X}\) and \(\delta \leq 1\) we obtain
\[
\Prp{X \geq (1+\delta) \mu} \leq \DepProb(\ssigma/2, d, \Sigma) + 4 \cdot \exp\ld(- \delta^2 \cdot \min\set{\ssigma/2, \mu}/3 \rd) \, .
\]

Let \(X_i\) be the size of the preimage of the long interval \(I_i\) of length \(\Delta_i\) with \(\mu_i = \alpha \Delta_i\).
Then
\begin{align*}
  \Prp{X_i \geq (1-\eps/3) \Delta_i} &= \Prp{X_i \geq \ld(1 + \frac{2 \eps}{3 \alpha}\rd) \alpha \Delta_i} \\
  &\leq \Prp{X_i \geq \ld(1 + \frac{2 \eps}{3}\rd) \mu_i} \\
  &\leq \DepProb(\ssigma/2, d, \Sigma) + 4 \cdot \exp\ld(- \ld(\frac{2\eps}{3}\rd)^2 \cdot \frac{\min\set{\mu_i, \ssigma/2}}{3} \rd)
\end{align*}
Notice how the probability is non-increasing for increasing sizes of the intervals. Thus we bound each of the probabilites for a long interval exceeding its threshold by the probability obtained for \(I_0\) with \(\mu_0 = \alpha \Delta \leq \ssigma/2\).

For start-interval \(I_i'\) of length \(\Delta_i'\) with \(\mu_i' = \alpha \Delta_i' = \Delta_i \cdot \eps/6\) we observe the same pattern
\begin{align*}
  \Prp{X_i' \geq \eps/3 \Delta_i} &= \Prp{X_i' \geq 2\mu_i'} \\
  &\leq \DepProb(\ssigma/2, d, \Sigma) + 4 \cdot \exp\ld(- \frac{\max\set{\mu_i', \ssigma/2}}{3}\rd) \, ,
\end{align*}
where we can bound the probability that each start-interval is too large by the probability obtained for \(I_0'\) with \(\mu_0' = \Delta \eps/6\).

The probability that \emph{any} of our intervals is too large is thus at most
\begin{align*}
  2 \log_{1 + \eps/(6\alpha)} m \cdot \ld( \DepProb(\ssigma/2, d, \Sigma) + 4 \exp\ld(- 4/27 \cdot \eps^2 \alpha \Delta\rd)\rd)
\end{align*}
where use that \(4/9 \cdot \eps^2 \alpha \Delta \leq \eps \Delta / 9 \leq \eps \Delta/6\) as \(\eps+\alpha=1\), hence the probability obtained for \(I_0\) is larger than that for \(I_0'\).

Let us rewrite this expression in terms of \(n\) and \(\ssigma\), our main parameters.
\begin{align*}
  \log_{1 + \eps/(6\alpha)} m &= \log_{1 + \eps/(6\alpha)} n + \log_{1 + \eps/(6\alpha)} (1/\alpha) \\
  &\leq \log n \cdot \log_{1 + \eps/(6\alpha)} (2) + 6 \\
  &\leq \log n \cdot 6 \alpha/\eps + 6 \\
  &\leq 30 \cdot \log n \, ,
\end{align*}
using \(\eps \geq 1/5\).
As \(\alpha \Delta \geq \ssigma/12\), we get
\begin{align*}
  4 \exp\ld(- 4/27 \cdot \eps^2 \alpha \Delta\rd) &\leq 4 \exp\ld(- 4/27 \cdot \eps^2 \ssigma/12\rd) \\
  &\leq 4 \exp\ld(- 1/2025 \cdot \ssigma\rd) \\
  &\leq \frac{1}{2\ssigma^2} \,
\end{align*}
as \(\ssigma \geq 2^{16}\).
Assuming \(30 \cdot \log n \leq \ssigma\) the total error-probability becomes
\[  1/(2\ssigma) + 30 \cdot \log n \cdot \DepProb(\ssigma/2, d, \Sigma) \, .\]

Repeating the process once more to ensure that the run at \(T[h(x)]\) doesn't continue past \(T[h(x) + \Delta]\) doubles the error-probability and proves the lemma.

\subsection{Proof of~\cref{cor:linearprobing}} \label{subsec:linearprobing-expectation}
To bound the expected number of comparisons we rely on the following lemma from \cite{patrascu12charhash} which gives strong concentration bounds for the runlength when applied to our simple tabulation \(\widehat h\).
\begin{lemma}[Corollary 3.2 in~\cite{patrascu12charhash}]\label{lemma:runlength}
 For any
  $\gamma=O(1)$ and $\ell\leq n^{1/(3(c+d))}/\alpha$,
\begin{align*}
\Pr[R(x,S)\geq \ell]\leq \left\{\begin{array}{ll}
2e^{-\Omega(\ell\eps^2)} + (\ell/m)^{\gamma} &\mbox{if }\alpha\geq 1/2\\
\alpha^{\Omega(\ell)} + (\ell/m)^{\gamma} &\mbox{if }\alpha\leq 1/2
\end{array}\right.
\end{align*}
where the constants hidden in \(O\) and \(\Omega\) are functions of \(c+d\), the size of the derived keys on which we apply simple tabulation.
\end{lemma}

In particular this implies that, for some $\ell = \Theta\ld((\log n)/\eps^2\rd)$, we have that $R(q,S)\geq \ell$ with probability at most $1/n^{10}$.
Let \(\calE\) be the event of stochastic dominance, as given by \cref{thm:linearprobing}, and \(\calA\) the event \((R(x, S) \leq \ell)\).
Then
\begin{align*}
  \Ep{X} &= \Epcond{X}{\calE} \cdot \Prp{\calE} + \Epcond{X}{\bar \calE \land \calA} \cdot \Prp{\bar \calE \land \calA} + \Epcond{X}{\bar \calE \land \bar \calA} \cdot \Prp{\bar \calE \land \bar \calA} \, .
\end{align*}
First, observe
\begin{align*}
  \Epcond{X}{\calE} \cdot \Prp{\calE} &= \sum_{i=1} \Prpcond{X \geq i}{\calE} \cdot \Prp{\calE} \\
  &\leq \sum_{i=1} \Prpcond{X^* \geq i}{\calE} \cdot \Prp{\calE} \\
  &= \sum_{i=1} \Prp{X^* \geq i \land \calE} \\
  &\leq \sum_{i=1} \Prp{X^* \geq i} \\
  &= \Ep{X^*} \, .
\end{align*}

With \(\delta = 1/\ssigma\) and \(d \geq 5\), \(\Prp{\bar \calE} \leq 9/\ssigma\).
Assuming \(\calA\), the next open cell of \(T\) is at most \(\ell\) positions away,
\begin{align*}
  \Epcond{X}{\bar \calE \land \calA} \cdot \Prp{\bar \calE \land \calA} &\leq \ell \cdot \Prp{\bar \calE} \\
  &\leq \Theta\ld(\frac{\log n}{\eps^2}\rd) \cdot \frac{9}{\ssigma} \\
  &\leq o(1) \, .
\end{align*}
Finally, no more than \(n\) comparisons will ever be necessary.
Hence
\begin{align*}
  \Epcond{X}{\bar \calE \land \bar \calA} \cdot \Prp{\bar \calE \land \bar \calA}
  &\leq n \cdot \Prp{\bar \calA} \\
  &\leq 1/n^{9} \, .
\end{align*}
This gives the desired bound on \(\Ep{X}\),
\begin{align*}
  \Ep{X} &\leq \Ep{X^*} + o(1) + 1/n^{9} \, .
\end{align*}

\section{Lower Bound for Tornado Tabulation}\label{sec:lower-bound}

In this section, we show that the probability obtained in~\Cref{thm:tech-random-set} is tight up to constant factors. Specifically, we will prove the following:
\lowerbound*
Our strategy will mimic that in the proof of~\Cref{thm:tech-random-set} and show that the set of derived selected keys will contain a zero-set with probability at least $\Theta((\mu^f)^3(3/|\Sigma|)^{d+1})$. We begin by establishing some initial general bounds. In the following, we define $\derive h':\Sigma^c\rightarrow \Sigma^{c+d}$ to map keys in $\Sigma^c$ to \emph{simple} derived keys in $\Sigma^{c+d}$ by applying the same functions as $\derive h$ except with $\derive h_0(\cdot)=0$, i.e., for all $i>1$,  $\derive h'_{c+i} = \derive h_{c+i}$.

\begin{definition} We say that a zero-set  $Y \subseteq  \Sigma^c$   \emph{survives} $d$  rounds of tornado tabulation if the set $\derive h' (Y)\subseteq \Sigma^{c+d}$ of its simple derived keys  is also a zero-set. 
\end{definition}

We focus on zero-sets of size $4$ and lower bound the probability that they survive successive rounds of tornado tabulation.  We first define some terminology necessary to describe how each new derived character in $\derive x'$ behaves. Specifically, let $Y = \set{x_1,x_2,x_3,x_4}$ be a zero-set, for some fixed ordering of its keys. We distinguish between four types of positions $i\in\set{1,\ldots, c}$ as such:
(1) position $i$ is of Type $A$ iff $x_1[i]=x_2[i]$ and $x_3[i]=x_4[i]$,~(2) it is of Type $B$ iff $x_1[i]=x_3[i]$, $x_2[i]=x_4[i]$,~(3) it is of Type $C$ iff $x_1[i] = x_4[i]$ and $x_2[i]=x_3[i]$ and,~(4) it is of Type $D$ iff $x_1[i] = x_2[i] = x_3[i]=x_4[i]$. We now prove that"

\begin{lemma}\label{lem:1roundsurv} Let $Y \subseteq \Sigma^c$ be a zero-set with $\size{Y}=4$. Then, for any $c\geq 2$, $Y$ survives one round of tornado tabulation with probability
	$\parentheses{3-2/\size{\Sigma}} /\size{\Sigma}\;.$
\end{lemma}
\begin{proof}
	Since the original keys in $Y$ already form a zero-set, the set of simple derived keys $\derive h'(Y)$ is a zero-set iff the set of simple derived characters $\derive h'_{c+1}(Y)$ is a zero-set.
	Moreover, the cases in which $\derive h'_{c+1}(Y)$ is a zero-set can  be classified based on the type of position $c+1$.  Specifically, let $\calA_{c+1}$ denote the event that position $c+1$ is of Type $A$, i.e., $\derive h'_{c+1}(x_1) = \derive h'_{c+1}(x_2)$ and $\derive h'_{c+1}(x_3) = \derive h'_{c+1}(x_4)$, and similarly for $\calB_{c+1}$, $\calC_{c+1}$, and $\calD_{c+1}$. Then
	\begin{align*}
		\Pr\parentheses{Y \text{ survives one round}} &= \Pr\parentheses{\derive h'_{c+1}(Y) \text{ is a zero-set }}\\
		&= \Pr\parentheses{\calA_{c+1} \vee \calB_{c+1}  \vee \calC_{c+1}}\\
		&= \Pr\parentheses{\calA_{c+1}} +  \Pr\parentheses{\calB_{c+1}} + \Pr\parentheses{\calC_{c+1}}  -  \Pr\parentheses{\calA_{c+1} \wedge \calB_{c+1}} - \\
		&  \Pr\parentheses{\calA_{c+1} \wedge \calC_{c+1}} -  \Pr\parentheses{\calB_{c+1} \wedge \calC_{c+1}}  +  \Pr\parentheses{\calA_{c+1} \wedge \calB_{c+1} \wedge \calC_{c+1}} \\
		& = 3\Pr\parentheses{\calA_{c+1}}  -  2\Pr\parentheses{\calD_{c+1}}\;,
	\end{align*}
	
	\medskip\noindent
	where the last equality follows from the fact that the events $\calA_{c+1}$, $\calB_{c+1}$ and $\calC_{c+1}$ are equivalent up to a permutation of the elements in $Y$, and the fact  that the conjunction of any pair of events in $\calA_{c+1}$, $\calB_{c+1}$ and $\calC_{c+1}$ implies $\calD_{c+1}$, and vice-versa. 
	
	\medskip\noindent
	We now bound $\Pr\parentheses{\calA_{c+1}}$ and $\Pr\parentheses{\calD_{c+1}}$.
	Recall that, by definition, the simple derived character $\derive h'_{c+1}(x)$ is the output of a simple tabulation hash function applied to the key $x$. Specifically, for each $i\in\set{1,\ldots, c}$, let $T_i:\Sigma\rightarrow\Sigma$ denote a fully random function. Then
	$$\derive h'_{c+1}(x) = T_1(x[1])\xor \ldots \xor T_c(x[c])\;.$$
	
	\medskip\noindent
	Let $I_A$, $I_B$, $I_C$ and $I_D$ partition  the set of positions in the original keys $\set{1,\ldots, c}$ based on their type, i.e., $I_A$ consists of positions that are of Type $A$ but not Type $D$, similarly for $I_B$ and $I_C$, and finally $I_D$ denotes the positions of Type $D$. 
	We then define $T_A(x) = \xor_{i\in I_A} T_i[x[i]]$ and similarly $T_B(x)$, $T_C(x)$, and $T_D(x)$.
	When $\calA_{c+1}$ happens, we have that $\derive h'_{c+1}(x_1) = \derive h'_{c+1}(x_2)$, which is equivalent to
	$$ T_B(x_1) \xor T_C(x_1)= T_B(x_2) \xor T_C(x_2)\;,$$
	since  $T_A(x_1) = T_A(x_2)$ and $T_D(x_1) = T_D(x_2)$ by definition.  
	Similarly, $\derive h'_{c+1}(x_3) = \derive h'_{c+1}(x_4)$, is equivalent to
	$$ T_B(x_3) \xor T_C(x_3)= T_B(x_4) \xor T_C(x_4)\;.$$
	Note that, by definition, $x_1[i] = x_3[i]$ and $x_2[i] = x_4[i]$ for all $i\in I_B$, and hence $T_B(x_1) = T_B(x_3)$ and $T_B(x_2) = T_B(x_4)$. Similarly, $T_C(x_1) = T_C(x_4)$ and $T_C(x_2) = T_C(x_3)$. Therefore, both equalities are equivalent to 
	$$ T_B(x_1) \xor T_C(x_1) \xor T_B(x_2) \xor T_C(x_2) = 0 \;.$$
	Given that $x_1[i] \neq x_2[i]$ for all $i\in I_B \cup I_C$ and the $T_i$'s are independent, we have that
	$$\Pr[T_B(x_1) \xor T_C(x_1) \xor T_B(x_2) \xor T_C(x_2) = 0] = 1/\size{\Sigma}\;.$$
	
	\medskip\noindent
	In order to bound $\Pr\parentheses{\calD_{c+1}}$,  we first note that
	$$ \Pr\parentheses{\calD_{c+1}} = \Pr\parentheses{\calA_{c+1} \wedge \calB_{c+1}}
	= \Pr\parentheses{\calA_{c+1}} \cdot \Pr\parentheses{\calB_{c+1} \mid \calA_{c+1}} = 1/\size{\Sigma} \cdot  \Pr\parentheses{\calB_{c+1} \mid \calA_{c+1}}\;.$$
	A similar argument as before shows that event $\calB_{c+1}$ is equivalent to
	$$ T_A(x_1) \xor T_C(x_1) \xor T_A(x_3) \xor T_C(x_3) = 0 \;.$$
	Note, in particular, that the event $\calB_{c+1}$ depends on positions in $I_A$ and $I_C$, while the event $\calA_{c+1}$ depends on positions in $I_B$ and $I_C$. Moreover, it cannot be that both $I_A$ and $I_B$ are empty, since then we would not have a zero-set of size $4$ (i.e., we would get that $x_1= x_4$ and $x_2 = x_3$). Therefore, $I_B \cup I_C \neq I_A \cup I_C$ and the two events $\calA_{c+1}$ and $\calB_{c+1}$ are independent, and so $\Pr\parentheses{\calD_{c+1}} = 1/\size{\Sigma}^2$. The claim follows.
\end{proof}

\medskip\noindent
As a corollary, we get the following:

\begin{corollary}\label{cor:droundsurv}For any $c\geq 2$, a zero-set  $Y \subseteq \Sigma^c$ with $\size{Y}=4$ survives $d$ rounds of tornado tabulation with probability
	$\parentheses{\parentheses{3-2/\size{\Sigma}} /\size{\Sigma}}^d\;.$
	
\end{corollary}
\begin{proof}
	We prove the claim by induction on $d$ and note that the case in which $d=1$ is covered in~\cref{{lem:1roundsurv}}. Now assume that the statement is true for $d-1$. 
	Recall that $\derive x'$ denotes the simple derived key and that $\derive x'[\leq c+d-1]$ denotes the first $c+d-1$ characters of $\derive x'$. By extension, let $\derive Y'$ denote the set of simple derived keys of $Y$ and similarly, $\derive Y'[\leq c+d-1] = \set{\derive x'[\leq c+d-1] \mid x\in Y}$ and $\derive Y'[c+d] = \set{\derive x'[c+d]\mid x\in Y}$.  Finally,  let $\calE_{d-1}$ and $\calE_{d}$ denote the events that the set $\derive Y'$  and $\derive Y'[\leq c+d-1]$, respectively, are zero-sets.  Then:
	\begin{align*}
		\Pr\parentheses{\calE_d} &= \Pr\parentheses{\calE_{d-1} \text{ and } \derive Y'[c+d] \text{ is a zero-set}} \\
		&= \Pr\parentheses{\calE_{d-1}} \cdot \Pr\parentheses{ \derive Y'[c+d] \text{ is a zero-set } \mid \calE_{d-1}} \\
		& =\parentheses{\parentheses{3-2/\size{\Sigma}} /\size{\Sigma}}^{d-1} \cdot  \Pr\parentheses{ \derive Y'[c+d] \text{ is a zero-set } \mid \calE_{d-1}} \;,
	\end{align*}
	where the last equality holds by the inductive hypothesis. To finish things up, we note that the event  $\derive Y'[c+d]$ conditioned on $\calE_{d-1}$ is equivalent to the set $\derive Y'[\leq c+d-1]$ surviving one round of tabulation hashing. Hence, it happens with probability $\parentheses{3-2/\size{\Sigma}} /\size{\Sigma}$ and the claim follows.
	
\end{proof}

\subsection{Proof of~\Cref{thm:lower-bound} }
\textbf{The hard instance.}
Consider the set of keys $S = \set{0,1}\times \Sigma$ and note that $\derive h_0$ on this set induces a permutation of the characters in $\Sigma$. Specifically, every key of the form $0c$ for some $c\in\Sigma$ will be mapped to the element $0c'$, where $c'= \derive h_0(0)\xor c$ and the mapping $c\rightarrow \derive h_0(0)\xor c$ is a permutation. Similarly for keys $1c$. Therefore, we can assume without loss of generality that $\derive h_0(\cdot) = 0$ and get that:
\begin{align*}
	\Pr\parentheses{\derive h(X^{f,h}) \text{ is linearly dep.}} &= \Pr\parentheses{\exists \text{ a four-set } Y \subseteq X^{f,h}\text{ that survives $d$ rounds of tornado tab.}}.
\end{align*}

We then define the selector function to select a key  $x$ if $x\in S$ and the two leftmost output characters of $h(x)$ are both $0$. Note then than that the probability that an $x\in S$ gets selected to $X^{f,h}$ is $1/4$ and hence, $\mu^f= \size{\Sigma/2}$.  We now focus on zero-sets from $S$ of size $4$ and,  for any $c_1,c_2 \in \Sigma$ with $c_1< c_2$, we denote the zero-set of size four $ \set{0c_1, 1c_1,0c_2, 1c_2}$ by  $Y(c_1,c_2)$.  We then let $\calY = \set{Y(c_1,c_2) \mid c_1<c_2 \in \Sigma}$ and let  $\calE_i(Y)$ denote the event that a zero-set $Y$ survives $i$ rounds of tornado tabulation. We lower bound the probability that $\derive h(X^{f,h})$ is linearly dependent by only focusing on zero-sets in $\calY$:

\begin{align*}
	\Pr\parentheses{\derive h(X^{f,h}) \text{ is linearly dep.}} &\geq \Pr\parentheses{\exists Y \in \calY \text{ s.t. } \calE_d(Y)  \wedge Y \subseteq X^{f,h}}\\
	&\geq \Pr\parentheses{\exists  Y\in \calY \text{ s.t. } \calE_d(Y)  \wedge Y \subseteq X^{f,h}}\\
	&\geq \Pr\parentheses{\exists  Y\in \calY \text{ s.t. } \calE_d(Y) } \cdot \Pr\parentheses{ \text{a fixed }Y\in \calY, Y \subseteq X^{f,h} \conditioned  \calE_d(Y) }\\
	&\geq 1/4^3 \cdot  \Pr\parentheses{\exists  Y\in \calY \text{ s.t. } \calE_d(Y) } \;,
\end{align*}

where the last two inequalities above are due to the fact that  the probability that some fixed set $Y\in \calY$ gets selected in $X^{f,h}$ given that it survived $d$ rounds of tabulation hashing is the same across all sets in $\calY$ and, furthermore, it is exactly $1/4^3$. 

\medskip\noindent
\textbf{Surviving zero-sets.} We now employ~\Cref{cor:droundsurv} to lower bound the probability that some zero-set in $\calY$ survives $d$ rounds of tornado tabulation. Note that  we already have that the expected number of zero-sets in $\calY$ that survive $d$ rounds of tornado tabulation is 
$$ \Theta(\size{\Sigma}^2 \cdot \parentheses{\parentheses{3-2/\size{\Sigma}} /\size{\Sigma}}^d) = \Omega((3/\size\Sigma)^{d-2})\;,$$
which exhibits the desired dependency on $d$. 
The challenge with turning this expectation into a probability is that the events of sets in $\calY$ surviving a round are not independent. To address this, we  decompose the event of some set $Y$ surviving $d$ rounds of tornado tabulation into the event that some set survives the first two rounds of tornado tabulation and the event that this set also survives the remaining $d-2$ rounds:

\begin{align*}
	\Pr\parentheses{\derive h(X^{f,h}) \text{ is linearly dep.}}  &\geq 1/4^3 \cdot  \Pr\parentheses{\exists  Y\in \calY \text{ s.t. } \calE_d(Y) } \\
	&\geq 1/4^3 \cdot  \Pr\parentheses{\exists  Y\in \calY \text{ s.t. } \calE_2(Y) } \cdot \Pr\parentheses{ \calE_{d-2}(\derive h(Y)[\leq c+2]) \conditioned \calE_2(Y)} \\
	&= 1/4^3 \cdot \parentheses{\parentheses{3-2/\size{\Sigma}} /\size{\Sigma}}^{d-2} \cdot \Pr\parentheses{\exists  Y\in \calY \text{ s.t. } \calE_2(Y) } \;.
\end{align*}

\medskip\noindent
To finish the argument and prove the main claim, we show the following:

\begin{lemma}
	With constant probability, at least one set in $\calY$ survives the first two rounds of tornado tabulation.
\end{lemma}
\begin{proof}

The proof proceeds in two stages: first, we will argue that, with constant probability, $\Theta(\size{\Sigma})$ of the sets in $\calY$ survive the first round of tornado tabulation. These sets will have a specific structure that guarantees that they then  survive the second round of tornado tabulation independently. 

Let $T^{(1)}_1, T^{(1)}_2:\Sigma \rightarrow \Sigma$ be the two fully random hash functions involved in computing the first derived character, and let
$\calC_1$ denote the event that $T^{(1)}_1(0) \neq T^{(1)}_1(1)$. Note that $\calC_1$ happens with probability $1-1/\size{\Sigma}$.
Conditioned on $\calC_1$, all the sets in $\calY$ that survive have position $3$ of Type $B$ or $C$.  For any $Y(c_1,c_2)\in\calY$ , position $3$ is of Type $B$ if $T^{(1)}_2(c_1) = T^{(1)}_2(c_2)$ and of Type $C$ if $T^{(1)}_2(c_2)= T^{(1)}_1(0) \xor T^{(1)}_1(1) \xor T^{(1)}_2(c_1)$.  These events are  mutually exclusive and each occurs with probability $1/\size{\Sigma}$. 

We now show that, with constant probability, at least $\Theta(\size{\Sigma)}$ of sets in $\calY$ will survive and futher, have position $3$ be of Type $B$. 
We model this as a balls-into-bins game in which there is a bin for each character $\alpha \in \Sigma$ and  the characters in $\Sigma$ hash into bins using $T^{(1)}_2$. We let $N_\alpha$ denote the number of characters $c \in\Sigma$ with $T^{(1)}_2(c_1)=\alpha$, i.e., the occupancy of the bin for $\alpha$. We are interested in events in which $N_\alpha \geq 2$, because this implies that there exists at least one set $Y(c_1,c_2)\in\calY$ where $T^{(1)}_2(c_1) = T^{(1)}_2(c_2)=\alpha$.
The probability that this occurs is:

\begin{align*}
	\Pr\parentheses{N_\alpha\geq 2} &= 1- \Pr\parentheses{N_\alpha=0}  - \Pr\parentheses{N_\alpha=1}= 1- (1-1/\size{\Sigma})^{\size{\Sigma}} - (1-1/\size{\Sigma})^{\size{\Sigma}-1} \;,
\end{align*} 
since each character hashes independently and uniformly into the bins. 
Now, for each $\alpha\in\Sigma$, define $I_\alpha$ to be the indicator random variable for whether $N_\alpha \geq 2$ and let $I = \sum_{\alpha\in\Sigma} I_{\alpha}$. We will argue that $I=\Theta(\Sigma)$ with constant probability. First note that the random variables $\set{ I_{\alpha}}_{\alpha\in\Sigma}$ are negatively associated since bin occupancies are negatively associated~\cite{dubhashi1998balls}. Let $\mu = \E(I)$ and note that $\mu \geq \size{\Sigma}/4$ for $\size{\Sigma}\geq 2$.  As such, we can apply Chernoff's bound and get that:
 
 \begin{align*}
 	\Pr\parentheses{I \leq \size{\Sigma}/ 8} \leq \Pr\parentheses{I\leq (1-1/2) \cdot \mu} \leq e^{-\mu/8} \leq 1/2\;,
 \end{align*}
where the last inequality holds when $\size{\Sigma} \geq 23$.  

\medskip\noindent
Now let $\calY'\subset\calY$ denote the set of zero-sets constructed as such: we partition the bins into subsets of the form $\set{\alpha,\alpha'}$ where $\alpha \xor \alpha' = T^{(1)}_1(0) \xor T^{(1)}_1(1)$. The event that $I \geq  \size{\Sigma}/ 8$ implies that there are at least $\size{\Sigma}/16$ such subsets where at least one bin, say $\alpha$, has $N_{\alpha} \geq 2$. Now let $c_1 < c_2$ be two such characters that hash into the bin and add the zero-set $Y(c_1,c_2)$ to $\calY'$. Note that every subset of bins contributes at most one zero-set to $\calY'$. Furthermore, the sets in $\calY'$ have the following properties:
\begin{itemize}
	\item  for any two distinct sets $Y(c_1,c_2) , Y(c_3,c_4)\in \calY'$, it holds that $\set{c_1,c_2} \cap \set{c_3,c_4} = \emptyset$, since the characters $c_1,c_2$ hash into a different bin from $c_3,c_4$,
	\item if we denote the derived characters  of $Y(c_1,c_2)$ by  $\derive h_{3}(0c_1) = x_1$ and $\derive h_{3}(1c_1) = x_2$ and similarly, the derived characters of $Y(c_3,c_4)$ by $\derive h_{3}(0c_3) = y_1$ and $\derive h_{3}(1c_3) = y_2$, we have further that $\set{x_1,x_2} \cap \set{y_1, y_2} = \emptyset$. This is due to the way the derived characters are computed: on one hand, $x_1 =  T^{(1)}_1(0) \xor  T^{(1)}_2(c_1) \neq T^{(1)}_1(0) \xor  T^{(1)}_2(c_3) = y_1$ because $T^{(1)}_2(c_1) \neq T^{(1)}_2(c_3) $ and similarly for $x_2 \neq y_2$. On the other hand, $x_1 \neq y_2$ because otherwise we would get that $T^{(1)}_2(c_1) \xor T^{(1)}_2(c_3) = T^{(1)}_1(0) \xor  T^{(1)}_1(1)$, which would contradict the fact that  $Y(c_1,c_2)$ and  $Y(c_3,c_4)$ were generated by different subsets of bins.
\end{itemize}

\medskip\noindent
In this context, the events in which sets in $\calY'$ survive the second round of tornado tabulation are independent.
Specifically, let $Y(c_1,c_2)\in\calY'$ be a set as before with derived characters $\derive h_{3}(0c_1) = \derive h_{3}(0c_2) = x_1$ and $\derive h_{3}(1c_1) = \derive h_{3}(1c_2)= x_2$. We distinguish between whether the newly derived character is of Type $A$, $B$, or $C$. To this end, let $T^{(2)}_1, T^{(2)}_2, T^{(2)}_3:\Sigma \rightarrow \Sigma$ be the fully random hash functions involved in its computation. Then position $4$ is of Type $A$ if 
\begin{align}\label{eq:typea}
	T^{(2)}_1 (0)  \xor T^{(2)}_3 (x_1) = T^{(2)}_1 (1)  \xor T^{(2)}_3 (x_2)\;.
\end{align}
Position $4$ is of Type $B$ if
\begin{align}\label{eq:typeb}
 T^{(2)}_2 (c_1)  = T^{(2)}_2 (c_2)\;,
 \end{align}
and of Type $C$ if
\begin{align}\label{eq:typec}
 T^{(2)}_1 (0)  \xor  T^{(2)}_2 (c_1)  \xor T^{(2)}_3 (x_1) = T^{(2)}_1 (1)  \xor T^{(2)}_2 (c_2) \xor T^{(2)}_3 (x_2)\;.
\end{align}
Similar conditions hold for some other $Y(c_3,c_4) \in \calY'$, and moreover, each of them depends on values that are chosen independently from the values in $Y(c_1,c_2)$.
Specifically, the analogues of~\Cref{eq:typea} for $Y(c_3,c_4)$ depends on the lookup table values of $y_1$ and $y_2$, where $y_1$ and $y_2$ are the derived characters  $\derive h_{3}(0c_3) = y_1$ and $\derive h_{3}(1c_3) = y_2$, respectively. As noted before, we know that $\set{x_1,x_2} \cap \set{y_1, y_2} = \emptyset$, and so the analogue of~\Cref{eq:typea} for $Y(c_3,c_4)$ is independent of Equations~(\ref{eq:typea}),~(\ref{eq:typeb}), and~(\ref{eq:typec}). Similar arguments can be made for the other cases.

Now fix some instantiation $Y'$ of $\calY'$ of size $\size{\Sigma}/16$ and let $X_{Y'}$ denote the number of sets in $Y'$ that survive the second round of tornado tabulation. We know from~\Cref{lem:1roundsurv} that each set in $Y'$ survives the second round of tornado tabulation with probability $(3-2/\size{\Sigma})/\size{\Sigma} \geq 2.75/16$ for $\size{\Sigma}\geq 8$. 
Furthermore, each set survives this second round independently from the others. It follows from Chernoff's inequality that $X$ is then tightly concentrated around its mean. In particular,
\begin{align*}
\Pr\parentheses{X_{Y'} \leq 0.01 \cdot 1/8 \conditioned \calC_1} &\leq \Pr\parentheses{X_{ Y'} \leq (1-0.99) \cdot \E[X_{\calY'}]  \conditioned \calC_1} \\
&\leq e^{- \E[X_{\calY'}] \cdot 0.99^2/2} \leq e^{-2.75\cdot  0.99^2/32} \leq e^{-0.08} \;.
\end{align*}
For our purposes, let $X_{\calY'}$ denote the random variable that counts the number of sets in $\calY'$ that survive the second round of tabulation hashing. Conditioned on the fact that $\size{\calY'} \geq \size{\Sigma}/16$, we can always pick an instantiation $Y'$ of $\calY'$ on which to use the above bound. We then get that:
\begin{align*}
	 \Pr\parentheses{X_{\calY'} \geq 0.01/8\conditioned \size{\calY'} \geq \size{\Sigma}/16  \wedge \calC_1} \geq 1-e^{-0.08} \;.
 \end{align*}

\medskip\noindent
To put it all together, let $I_{\calY}$ denote the event that there exists $Y(c_1,c_2)\in\calY$ such that $Z_2(c_1,c_2)$. Recall that we defined the event $\calC_1$ to be that $T^{(1)}_1(0) \neq T^{(1)}_1(1)$. Then:
\begin{align*}
	\Pr\parentheses{I_{\calY}} &\geq \Pr\parentheses{I_{\calY}  \wedge\calC_1} \\
	&= \parentheses{1-\frac{1}{\size{\Sigma}}}\cdot \Pr\parentheses{I_{\calY} \conditioned \calC_1}\\
	&\geq  \parentheses{1-\frac{1}{\size{\Sigma}}}\cdot  \Pr\parentheses{I_{\calY}\wedge \size{\calY'} \geq \size{\Sigma}/16  \conditioned \calC_1}\\
	&\geq  \parentheses{1-\frac{1}{\size{\Sigma}}}\cdot  \frac{1}{2} \cdot \Pr\parentheses{I_{\calY}\conditioned \size{\calY'} \geq \size{\Sigma}/16  \wedge \calC_1}\\
	&\geq \parentheses{1-\frac{1}{\size{\Sigma}}}\cdot  \frac{1}{2} \cdot  \Pr\parentheses{X_{\calY'} \geq 1\conditioned \size{\calY'} \geq \size{\Sigma}/16  \wedge \calC_1}\\
	&\geq \parentheses{1-\frac{1}{\size{\Sigma}}}\cdot  \frac{1}{2} \cdot (1-e^{-0.08}) \;.
\end{align*}

\end{proof}

\bibliographystyle{acm}
\bibliography{general}

\begin{thebibliography}{10}

\bibitem{Aamand0KKRT22}
{\sc Aamand, A., Das, D., Kipouridis, E., Knudsen, J. B.~T., Rasmussen, P.
  M.~R., and Thorup, M.}
\newblock No repetition: Fast and reliable sampling with highly concentrated
  hashing.
\newblock {\em Proc. {VLDB} Endow. 15}, 13 (2022), 3989--4001.

\bibitem{aamand2020fast}
{\sc Aamand, A., Knudsen, J. B.~T., Knudsen, M. B.~T., Rasmussen, P. M.~R., and
  Thorup, M.}
\newblock Fast hashing with strong concentration bounds.
\newblock In {\em Proceedings of the 52nd Annual ACM SIGACT Symposium on Theory
  of Computing\/} (2020), pp.~1265--1278.

\bibitem{AamandKT21:dynamic-load}
{\sc Aamand, A., Knudsen, J. B.~T., and Thorup, M.}
\newblock Load balancing with dynamic set of balls and bins.
\newblock In {\em {STOC} '21: 53rd Annual {ACM} {SIGACT} Symposium on Theory of
  Computing, Virtual Event, Italy, June 21-25, 2021\/} (2021), S.~Khuller and
  V.~V. Williams, Eds., {ACM}, pp.~1262--1275.

\bibitem{amble1974ordered}
{\sc Amble, O., and Knuth, D.~E.}
\newblock Ordered hash tables.
\newblock {\em The Computer Journal 17}, 2 (1974), 135--142.

\bibitem{DBLP:conf/focs/ArbitmanNS10}
{\sc Arbitman, Y., Naor, M., and Segev, G.}
\newblock Backyard cuckoo hashing: Constant worst-case operations with a
  succinct representation.
\newblock In {\em 51th Annual {IEEE} Symposium on Foundations of Computer
  Science, {FOCS} 2010, October 23-26, 2010, Las Vegas, Nevada, {USA}\/}
  (2010), {IEEE} Computer Society, pp.~787--796.

\bibitem{DBLP:conf/focs/BenderKK21}
{\sc Bender, M.~A., Kuszmaul, B.~C., and Kuszmaul, W.}
\newblock Linear probing revisited: Tombstones mark the demise of primary
  clustering.
\newblock In {\em 62nd {IEEE} Annual Symposium on Foundations of Computer
  Science, {FOCS} 2021, Denver, CO, USA, February 7-10, 2022\/} (2021), {IEEE},
  pp.~1171--1182.

\bibitem{black98linprobe}
{\sc Black, J.~R., Martel, C.~U., and Qi, H.}
\newblock Graph and hashing algorithms for modern architectures: Design and
  performance.
\newblock In {\em Proc. 2nd International Workshop on Algorithm Engineering
  (WAE)\/} (1998), pp.~37--48.

\bibitem{celis1985robin}
{\sc Celis, P., Larson, P.-A., and Munro, J.~I.}
\newblock Robin hood hashing.
\newblock In {\em 26th Annual Symposium on Foundations of Computer Science
  (sfcs 1985)\/} (1985), IEEE, pp.~281--288.

\bibitem{christiani15indep}
{\sc Christiani, T., Pagh, R., and Thorup, M.}
\newblock From independence to expansion and back again.
\newblock To appear, 2015.

\bibitem{dahlgaard15k-partitions}
{\sc Dahlgaard, S., Knudsen, M. B.~T., Rotenberg, E., and Thorup, M.}
\newblock Hashing for statistics over k-partitions.
\newblock In {\em 2015 IEEE 56th Annual Symposium on Foundations of Computer
  Science\/} (2015), IEEE, pp.~1292--1310.

\bibitem{DahlgaardKT17:nips}
{\sc Dahlgaard, S., Knudsen, M. B.~T., and Thorup, M.}
\newblock Practical hash functions for similarity estimation and dimensionality
  reduction.
\newblock In {\em Advances in Neural Information Processing Systems 30: Annual
  Conference on Neural Information Processing Systems 2017, December 4-9, 2017,
  Long Beach, CA, {USA}\/} (2017), I.~Guyon, U.~von Luxburg, S.~Bengio, H.~M.
  Wallach, R.~Fergus, S.~V.~N. Vishwanathan, and R.~Garnett, Eds.,
  pp.~6615--6625.

\bibitem{DietzfelbingerH90}
{\sc Dietzfelbinger, M., and auf~der Heide, F.~M.}
\newblock A new universal class of hash functions and dynamic hashing in real
  time.
\newblock In {\em Automata, Languages and Programming, 17th International
  Colloquium, ICALP90, Warwick University, England, UK, July 16-20, 1990,
  Proceedings\/} (1990), M.~Paterson, Ed., vol.~443 of {\em Lecture Notes in
  Computer Science}, Springer, pp.~6--19.

\bibitem{dietzfel09splitting}
{\sc Dietzfelbinger, M., and Rink, M.}
\newblock Applications of a splitting trick.
\newblock In {\em Proc. 36th International Colloquium on Automata, Languages
  and Programming (ICALP)\/} (2009), pp.~354--365.

\bibitem{DIETZFELBINGER200747}
{\sc Dietzfelbinger, M., and Weidling, C.}
\newblock Balanced allocation and dictionaries with tightly packed constant
  size bins.
\newblock {\em Theoretical Computer Science 380}, 1 (2007), 47--68.
\newblock Automata, Languages and Programming.

\bibitem{dietzfel03tabhash}
{\sc Dietzfelbinger, M., and Woelfel, P.}
\newblock Almost random graphs with simple hash functions.
\newblock In {\em Proc. 25th ACM Symposium on Theory of Computing (STOC)\/}
  (2003), pp.~629--638.

\bibitem{10.1145/780542.780634}
{\sc Dietzfelbinger, M., and Woelfel, P.}
\newblock Almost random graphs with simple hash functions.
\newblock In {\em Proceedings of the Thirty-Fifth Annual ACM Symposium on
  Theory of Computing\/} (New York, NY, USA, 2003), STOC '03, Association for
  Computing Machinery, p.~629–638.

\bibitem{dubhashi1998balls}
{\sc Dubhashi, D., and Ranjan, D.}
\newblock Balls and bins: A study in negative dependence.
\newblock {\em Random Structures \& Algorithms 13}, 5 (1998), 99--124.

\bibitem{Flajolet07hyperloglog}
{\sc Flajolet, P., Éric Fusy, Gandouet, O., and Meunier, F.}
\newblock Hyperloglog: The analysis of a near-optimal cardinality estimation
  algorithm.
\newblock In {\em In Analysis of Algorithms (AOFA)\/} (2007).

\bibitem{DBLP:journals/mst/FotakisPSS05}
{\sc Fotakis, D., Pagh, R., Sanders, P., and Spirakis, P.~G.}
\newblock Space efficient hash tables with worst case constant access time.
\newblock {\em Theory Comput. Syst. 38}, 2 (2005), 229--248.

\bibitem{greenberg2014tight}
{\sc Greenberg, S., and Mohri, M.}
\newblock Tight lower bound on the probability of a binomial exceeding its
  expectation.
\newblock {\em Statistics \& Probability Letters 86\/} (2014), 91--98.

\bibitem{heileman05linprobe}
{\sc Heileman, G.~L., and Luo, W.}
\newblock How caching affects hashing.
\newblock In {\em Proc. 7th Workshop on Algorithm Engineering and Experiments
  (ALENEX)\/} (2005), p.~141–154.

\bibitem{HouenT22:chaos}
{\sc Houen, J. B.~T., and Thorup, M.}
\newblock Understanding the moments of tabulation hashing via chaoses.
\newblock In {\em 49th International Colloquium on Automata, Languages, and
  Programming, {ICALP} 2022, July 4-8, 2022, Paris, France\/} (2022),
  M.~Bojanczyk, E.~Merelli, and D.~P. Woodruff, Eds., vol.~229 of {\em LIPIcs},
  Schloss Dagstuhl - Leibniz-Zentrum f{\"{u}}r Informatik, pp.~74:1--74:19.

\bibitem{KW12}
{\sc Klassen, T.~Q., and Woelfel, P.}
\newblock Independence of tabulation-based hash classes.
\newblock In {\em Proc. 10th Latin American Theoretical Informatics (LATIN)\/}
  (2012), pp.~506--517.

\bibitem{knudsen2016linear}
{\sc Knudsen, M. B.~T.}
\newblock Linear hashing is awesome.
\newblock In {\em 2016 IEEE 57th Annual Symposium on Foundations of Computer
  Science (FOCS)\/} (2016), IEEE, pp.~345--352.

\bibitem{knuth63linprobe}
{\sc Knuth, D.~E.}
\newblock Notes on open addressing.
\newblock Unpublished memorandum. See
  \url{http://citeseer.ist.psu.edu/knuth63notes.html}, 1963.

\bibitem{li12oneperm}
{\sc Li, P., Owen, A.~B., and Zhang, C.-H.}
\newblock One permutation hashing.
\newblock In {\em Proc. 26thAdvances in Neural Information Processing
  Systems\/} (2012), pp.~3122--3130.

\bibitem{mitzenmacher08hash}
{\sc Mitzenmacher, M., and Vadhan, S.~P.}
\newblock Why simple hash functions work: exploiting the entropy in a data
  stream.
\newblock In {\em Proc. 19th ACM/SIAM Symposium on Discrete Algorithms
  (SODA)\/} (2008), pp.~746--755.

\bibitem{PP08}
{\sc Pagh, A., and Pagh, R.}
\newblock Uniform hashing in constant time and optimal space.
\newblock {\em SIAM J. Comput. 38}, 1 (2008), 85--96.

\bibitem{pagh07linprobe}
{\sc Pagh, A., Pagh, R., and Ru{\v z}i{\'c}, M.}
\newblock Linear probing with constant independence.
\newblock {\em SIAM Journal on Computing 39}, 3 (2009), 1107--1120.
\newblock See also STOC'07.

\bibitem{patrascu10kwise-lb}
{\sc P{\v a}tra{\c s}cu, M., and Thorup, M.}
\newblock On the $k$-independence required by linear probing and minwise
  independence.
\newblock In {\em Proc. 37th International Colloquium on Automata, Languages
  and Programming (ICALP)\/} (2010), pp.~715--726.

\bibitem{patrascu12charhash}
{\sc P{\v a}tra{\c s}cu, M., and Thorup, M.}
\newblock The power of simple tabulation-based hashing.
\newblock {\em Journal of the ACM 59}, 3 (2012), Article 14.
\newblock Announced at STOC'11.

\bibitem{PT13:twist}
{\sc P\v{a}tra\c{s}cu, M., and Thorup, M.}
\newblock Twisted tabulation hashing.
\newblock In {\em Proc. 24th ACM/SIAM Symposium on Discrete Algorithms
  (SODA)\/} (2013), pp.~209--228.

\bibitem{siegel04hash}
{\sc Siegel, A.}
\newblock On universal classes of extremely random constant-time hash
  functions.
\newblock {\em SIAM Journal on Computing 33}, 3 (2004), 505--543.
\newblock See also FOCS'89.

\bibitem{thorup11timerev}
{\sc Thorup, M.}
\newblock Timeouts with time-reversed linear probing.
\newblock In {\em Proc. IEEE INFOCOM\/} (2011), pp.~166--170.

\bibitem{thorup13doubletab}
{\sc Thorup, M.}
\newblock Simple tabulation, fast expanders, double tabulation, and high
  independence.
\newblock In {\em FOCS\/} (2013), pp.~90--99.

\bibitem{thorup12kwise}
{\sc Thorup, M., and Zhang, Y.}
\newblock Tabulation-based 5-independent hashing with applications to linear
  probing and second moment estimation.
\newblock {\em SIAM Journal on Computing 41}, 2 (2012), 293--331.
\newblock Announced at SODA'04 and ALENEX'10.

\bibitem{wegman81kwise}
{\sc Wegman, M.~N., and Carter, L.}
\newblock New classes and applications of hash functions.
\newblock {\em Journal of Computer and System Sciences 22}, 3 (1981), 265--279.
\newblock See also FOCS'79.

\bibitem{zobrist70hashing}
{\sc Zobrist, A.~L.}
\newblock A new hashing method with application for game playing.
\newblock Tech. Rep.~88, Computer Sciences Department, University of Wisconsin,
  Madison, Wisconsin, 1970.

\end{thebibliography}

\end{document}